\documentclass[11pt,a4paper]{article}
\usepackage[margin = 1in]{geometry}

\usepackage{amsmath,amssymb,amsthm,hyperref,color}
\usepackage{float}
\usepackage[font=small]{caption}
\usepackage[font=footnotesize]{subcaption}
\usepackage[inline]{enumitem}
\usepackage{filecontents}
\usepackage{tikz}
\hypersetup{colorlinks=true,citecolor=blue, linkcolor=blue,
urlcolor=blue}

\newtheorem{theorem}{Theorem}
\newtheorem{lemma}{Lemma}

\newtheorem{corollary}{Corollary}
\newtheorem{remark}{Remark}
\newtheorem{definition}{Definition}

\newtheorem*{lemclosedformula}{Lemma \ref{lem:closedformula}}
\newtheorem*{lempath}{Lemma \ref{lem:path}}
\newtheorem*{lemzero}{Lemma \ref{lem:zero}}
\newtheorem*{lemsmalllambda}{Lemma \ref{lem:smalllambda}}

\newtheorem*{thmprecision}{Theorem \ref{thm:precision}}

\def\NP{\mathsf{NP}}
\def\Zin{Z^{\mathsf{in}}}
\def\Zout{Z^{\mathsf{out}}}
\def\Reals{\mathbb{R}}
\def\Rnz{\mathbb{R}_{\neq0}}

\def\lambdab{\ensuremath{\boldsymbol{\lambda}}}
\def\Mb{\ensuremath{\mathbf{M}}}
\def\Bc{\ensuremath{\mathcal{B}}}

\def\Hardcore{\#\mathsf{BipHardCore}(\lambda,\Delta,c)}
\def\HardCore#1{\#\mathsf{BipHardCore}(\lambda,\Delta,#1)}
\def\TwoSpin#1#2{\#\mathsf{2Spin}(#1,#2,c)}

\makeatletter
\def\prob#1#2#3{\goodbreak\begin{list}{}{\labelwidth\z@ \itemindent-\leftmargin
                        \itemsep\z@  \topsep6\p@\@plus6\p@
                        \let\makelabel\descriptionlabel}
                \item[\it Name]#1
               \item[\it Instance]                #2
                \item[\it Output]#3
                \end{list}}
\makeatother

\def\bezakova{Bez\'{a}kov\'{a}}

\title{Implementations and the independent set polynomial below the Shearer threshold\thanks{To appear in \emph{Theoretical Computer Science}. A preliminary short version of this paper (without the proofs) appeared in the Proceedings
of ICALP 2017 \cite{confversion}.}}
\author{
Andreas Galanis\thanks{
  Department of Computer Science, University of Oxford, Wolfson Building, Parks Road, Oxford, OX1~3QD, UK.
  The research leading to these results has received funding from the European Research Council under
  the European Union's Seventh Framework Programme (FP7/2007-2013) ERC grant agreement no.\ 334828. The paper
  reflects only the authors' views and not the views of the ERC or the European Commission.
  The European Union is not liable for any use that may be made of the information contained therein.}
  \and
  Leslie Ann Goldberg$^\dag$
\and
 Daniel \v{S}tefankovi\v{c}\thanks{
Department of Computer Science, University of Rochester,
Rochester, NY 14627.  Research
supported by NSF grant CCF-0910415.}
 }

\date{}

\begin{document}

\maketitle
\begin{abstract}  
 The independent set polynomial is important in many areas of combinatorics, computer science, and statistical physics.
 For every integer $\Delta\geq 2$, the Shearer threshold is the value 
 $\lambda^*(\Delta)=(\Delta-1)^{\Delta-1}/\Delta^{\Delta}$ . It is known that for $\lambda < - \lambda^*(\Delta)$, there
 are graphs~$G$ with maximum degree~$\Delta$ whose independent set polynomial, evaluated at~$\lambda$, is at most~$0$.
 Also, there are no such graphs for any  $\lambda > -\lambda^*(\Delta)$. 
 This paper is motivated by the computational problem of approximating the independent
 set polynomial when $\lambda < - \lambda^*(\Delta)$.
 The key issue in  complexity bounds for this problem  is ``implementation''.
  Informally, an implementation of a real number $\lambda'$ is a graph whose hard-core partition function, evaluated at~$\lambda$,
simulates a vertex-weight of~$\lambda'$ in the sense that 
$\lambda'$ is the ratio between the contribution to the partition
function from independent sets containing a certain vertex and the contribution from independent sets that do not contain that vertex.
 Implementations are the cornerstone of  intractability results for the  problem of approximately evaluating the independent set polynomial.
 Our main result is that, for any $\lambda < - \lambda^*(\Delta)$, it is possible to implement a set of values that is dense over the reals.
 The result is tight in the sense that it
 is not possible to implement a set of values that is dense over the reals
 for any $\lambda> \lambda^*(\Delta)$. 
 Our result has already been used  in  a paper  with \bezakova{} (STOC 2018)
 to show that it is \#P-hard to approximate the evaluation of the independent set polynomial
 on graphs of degree at most~$\Delta$
 at any value $\lambda<-\lambda^*(\Delta)$. 
 In the appendix, we give an additional incomparable inapproximability result (strengthening
 the inapproximability bound to an exponential factor, but weakening the hardness to NP-hardness).
 \end{abstract}

\section{Introduction}

The independent set polynomial  is a fundamental object in computer science which has been studied with various motivations. From an algorithmic viewpoint, the evaluation of this  polynomial is crucial for  determining the applicability of the Lov\'{a}sz Local Lemma and thus obtaining efficient algorithms for
both
finding \cite{MoserTardos} and approximately counting \cite{Guo,Moitra} combinatorial objects with specific properties.

The independent set polynomial also arises in statistical physics, where it is
called the \emph{hard-core partition function}.
Given a graph~$G$,
the value of the independent set polynomial
of~$G$
at a point $\lambda$
is equal to the value of the partition function of the hard-core model
where the so-called ``activity parameter'' is equal to~$\lambda$.
We use the following notation.
Given a graph~$G$,   let $\mathcal{I}_G$ denote the set of independent sets in $G$.
The weight of an independent set $I\in \mathcal{I}_G$ is given by $\lambda^{|I|}$.
The hard-core   partition function with parameter $\lambda$ is defined as 
\begin{equation}\label{eq:defZG}
Z_G(\lambda):=\sum_{I\in \mathcal{I}_G}\lambda^{|I|}.
\end{equation}

The evaluation of the hard-core partition function for $\lambda<0$ has significant algorithmic interest due to its connection with the Lov\'{a}sz Local Lemma (LLL) and, more precisely, to the problem of checking when the LLL applies. 
Shearer, as part of his work \cite{Shearer} on the LLL, implicitly established that for every $\Delta\geq 2$, there is a threshold $\lambda^*(\Delta)$, given by $\lambda^*(\Delta)=(\Delta-1)^{\Delta-1}/\Delta^{\Delta}$, such that 
\begin{enumerate}
\item for all $\lambda\geq-\lambda^*(\Delta)$, for all graphs $G$ of maximum degree $\Delta$, it holds that $Z_G(\lambda)>0$.
\item for all $\lambda<-\lambda^*(\Delta)$, there exists a graph $G$ of maximum degree $\Delta$  such that $Z_G(\lambda)\leq 0$.
\end{enumerate}
We refer to the point $-\lambda^*(\Delta)$ as the \emph{Shearer threshold}.
 Scott and Sokal \cite{ScottSokal} were the first to realise the relevance of  Shearer's work to  the phase transitions of the hard-core model, and to make explicit Shearer's contribution in this context. 
From an algorithmic viewpoint, the Shearer threshold is tacitly present in most, if not all, applications of the (symmetric) LLL. In particular, Shearer \cite{Shearer} proved that $\lambda^*(\Delta)$ is the maximum value $p$  such that 
in a probability space where each event occurs with probability
at most~$p$
and each event is independent of all    except at most~$\Delta$ other events,
there is a positive probability that no events occur.

A key concept in the hard-core model is that of ``implementation''. 
A formal definition follows shortly. 
 Informally, an \emph{implementation} of a real number $\lambda'$ is a graph whose hard-core partition function~$Z_G(\lambda)$
  simulates a vertex-weight of~$\lambda'$ in the sense that 
$\lambda'$ is the ratio between the contribution to~$Z_G(\lambda)$   from independent sets containing a certain vertex and the contribution from independent sets that do not contain that vertex.
Implementation is the cornerstone of all inapproximability results for the independence polynomial/hard-core partition function.

The main result of this paper (Theorem~\ref{thm:precision}) is that
if $\lambda<-\lambda^*(\Delta)$ then, for any real number~$\lambda'$
and any desired error, there is a bipartite graph~$G$ of maximum degree at most~$\Delta$ that 
implements the activity~$\lambda'$ within the desired error.
 
We first give the necessary definitions, in order to state the result formally. We then describe the algorithmic consequences.

\subsection{Implementations}

Let $\lambda\in \mathbb{R}$ and let~$G=(V,E)$ be a  graph. Given a vertex $v\in V$,  define
\[\Zin_{G,v}(\lambda):=\sum_{I\in \mathcal{I}_G;\, v\in I}\lambda^{|I|} \quad \mbox{and}\quad  \Zout_{G,v}(\lambda):=\sum_{I\in \mathcal{I}_G;\, v\notin I}\lambda^{|I|}.\]
Thus, $\Zin_{G,v}(\lambda)$ is the contribution to the partition function $Z_G(\lambda)$ from those independent sets $I\in \mathcal{I}_G$ such that $v\in I$; similarly, $\Zout_{G,v}(\lambda)$ is the contribution to $Z_G(\lambda)$ from those $I\in \mathcal{I}_G$ such that $v\notin I$.  

\begin{definition}\label{def:Gimplement} [Implementing an activity~$\lambda'$]
Fix $\lambda\in\Rnz$. We say that the graph $G$ implements the activity $\lambda'\in \mathbb{R}$ with error $\epsilon>0$ if there is a vertex $v$ in $G$ such that $\Zout_{G,v}(\lambda)\neq 0$  and
\begin{enumerate}
\item $v$ has degree one in $G$, and
\item  \label{it:epslambda12}  $\displaystyle \Big|\frac{\Zin_{G,v}(\lambda)}{\Zout_{G,v}(\lambda)}-\lambda'\Bigr| \leq \epsilon $.
\end{enumerate}
We  refer to the vertex $v$ as the ``terminal'' of $G$. \end{definition}

Suppose that $G$ implements~$\lambda'$ with error~$0$ and that~$v$ is the terminal of~$G$.
It is clear from the definition 
that $\lambda'$ is the ratio between $\Zin_{G,v}(\lambda)$ and $\Zout_{G,v}(\lambda)$.
To make use of this fact, reductions use the graph~$G$ as a ``gadget'' to simulate the activity~$\lambda'$.
This technique is crucial in all inapproximability results for the hard-core partition function --- the key to
showing inapproximability for a fixed activity~$\lambda$ is to use~$\lambda$ to implement a dense
set of activities~$\lambda'$.
Our main result shows that this is possible for every~$\lambda$ below the Shearer threshold.

\newcommand{\statethmprecision}{Let $\Delta\geq 3$ and $\lambda<-\lambda^*(\Delta)$.  Then, for every $\lambda'\in\mathbb{R}$, for every $\epsilon>0$, there exists a bipartite graph $G$ of maximum degree at most  $\Delta$ that implements $\lambda'$ with error $\epsilon$. }
\begin{theorem}\label{thm:precision}
\statethmprecision
\end{theorem}

Theorem~\ref{thm:precision} provides a sharp threshold in the sense that the theorem would be false for
any   $\lambda>-\lambda^*(\Delta)$. 
 In particular, in the regime $\lambda>-\lambda^*(\Delta)$, 
Scott and Sokal \cite{ScottSokal} have shown that $\Zin_{G,v}(\lambda)/\Zout_{G,v}(\lambda)>-1$ for all graphs $G$ of maximum degree $\Delta$ (and all vertices $v$ in $G$).

\subsection{Algorithmic Consequences} \label{sec:alg}

Our Theorem~\ref{thm:precision} is 
a key ingredient in the work of \cite{OurComplex} which demonstrates  the intractability of approximating the independent set polynomial  all the way up to the
Shearer threshold.  

We first explain the result of~\cite{OurComplex} and then explain how it relies on   our work.
Taken together, Theorem~1 and Corollary~3 of~\cite{OurComplex} show that, for any $\Delta\geq 3$ and $\lambda<-\lambda^*(\Delta)$,
it is \#P-hard to approximate the absolute value of $Z_G(\lambda)$ within a factor of~$1.01$, even when the input graph~$G$,
which has maximum degree at most~$\Delta$,
is restricted to be bipartite.\footnote{The paper \cite{OurComplex} also considered the case where~$\lambda$ is a complex number, but that is
beyond the scope of this paper.}
They show that, with the same restriction,
it is also \#P-hard to determine whether $Z_G(\lambda)$ is positive.
As \cite{OurComplex} explains, the value ``$1.01$'' in the statement of their Theorem~1 is not important --- for any fixed $\epsilon>0$, 
Theorem~1 can be  ``powered up'' to show 
that it is \#P-hard to approximate~$|Z_G(\lambda)|$ within a factor of~$2^{n^{1-\epsilon}}$, where $n$ is the number of vertices of~$G$.

The role of  our current work in the result of~\cite{OurComplex} is that our Theorem, Theorem~\ref{thm:precision}, is  re-stated (and used) as Lemma~16 of~\cite{OurComplex}, and it is a key part of the proofs of those theorems.
 
 One weakness of the results of~\cite{OurComplex} is that it is not known how to strengthen their theorems 
  to show that it is \#P-hard to  
approximate~$|Z_G(\lambda)|$ \emph{within an exponential factor}. 
However, this turns out to be the case, and is a result of our current  work. 
 In the appendix of this paper, we again use our main result, Theorem~\ref{thm:precision},
to give another inapproximability result 
(Theorem~\ref{thm:neglambda}) which shows that,
with the same restrictions on~$\Delta$, $\lambda$ and~$G$, it is
NP-hard to approximate~$|Z_G(\lambda)|$ within an exponential factor.

Before turning to the proof of our main theorem,
we  briefly give the context of the algorithmic consequences.
In addition to the Shearer threshold $\lambda^*(\Delta)$, there is another key threshold --- the threshold of the uniqueness phase transition, which is
given by 
$\lambda_c(\Delta):=(\Delta-1)^{\Delta-1}/(\Delta-2)^{\Delta}$.

We have already noted that it is  \#P-hard to approximate the absolute value of $Z_G(\lambda)$ 
for $\lambda<-\lambda^*(\Delta)$. The following complementary results 
give a full characterisation of the complexity of this problem for  real activities~$\lambda$,
apart from at the critical values $-\lambda^*(\Delta)$ and $\lambda_c(\Delta)$.
 
\begin{enumerate}
\item For $-\lambda^*(\Delta)<\lambda<\lambda_c(\Delta)$, there is an FPTAS for approximating $Z_G(\lambda)$ on graphs $G$ of maximum degree $\Delta$; this follows by  \cite{Piyush,PR} for $-\lambda^*(\Delta)<\lambda<0$ 
(see also \cite{Piyushnote})
and by \cite{Weitz} for $0<\lambda<\lambda_c(\Delta)$.\footnote{The case $\lambda=0$ is trivial since $Z_G(\lambda)=1$ for all graphs $G$.}
\item For  $\lambda>\lambda_c(\Delta)$, it is $\NP$-hard to approximate $|Z_G(\lambda)|$ on graphs $G$ of maximum degree $\Delta$, even within an exponential factor; this follows by   \cite{SlySun}.
\end{enumerate}

 \section{Proof of  the theorem}\label{sec:precision}

We start with some useful definitions. 
If a graph~$G$ implements an activity $\lambda' \in \mathbb{R}$ with error~$0$
then we say that \emph{$G$ implements~$\lambda'$}.
We use the following definition to avoid explicit mention of the graph~$G$.

\begin{definition}\label{def:implement}
Let $\Delta\geq 2$ be an integer and $\lambda\in \Rnz$. We say that $(\Delta,\lambda)$ \emph{implements} the activity $\lambda'\in\mathbb{R}$ if there is a \emph{bipartite} graph $G$ of maximum degree at most  $\Delta$ which implements the activity $\lambda'$.
More generally, we say that $(\Delta,\lambda)$ implements a set of activities $S\subseteq \mathbb{R}$, if for every $\lambda'\in S$ it holds that $(\Delta,\lambda)$ implements $\lambda'$.
\end{definition}

Using this notation, Theorem~\ref{thm:precision} says
is $\Delta \geq 3$ and $\lambda<-\lambda^*(\Delta)$
then $(\Delta,\lambda)$ implements a set of activities $S$ which is dense in $\mathbb{R}$.

\subsection{Overview of the proof}\label{sec:overviewprecision}
In this section, we give an overview of the proof of Theorem~\ref{thm:precision}. We also
state the main lemmas that we need, and use them to prove
 the theorem. The remaining sections  contain a proof of the main lemmas.

Consider a degree bound $\Delta\geq 3$. 
In order to prove Theorem~\ref{thm:precision}, we will  show that for any fixed $\lambda<-\lambda^*(\Delta)$, we can implement a dense set of activities using bipartite graphs of maximum degree $\Delta$. At a very rough level, the proof of Theorem~\ref{thm:precision} splits into two regimes:
\begin{enumerate}
\item \label{it:regime1} when $\lambda<-\lambda^*(2)=-1/4$,
\item \label{it:regime2} when $-1/4\leq \lambda<-\lambda^*(\Delta)$.
\end{enumerate}
Roughly, in regime~\ref{it:regime1}, we will be able to use \emph{paths} to implement a dense set of activities. In regime~\ref{it:regime2}, we will first use a $(\Delta-1)$-ary tree to implement an activity $\lambda'<-1/4$. Then, using the activity $\lambda'$, we will be able to use the path construction of the first regime to implement a dense set of activities. 

Unfortunately, the actual proof is more intricate, since as it turns out there is a set $\Bc\subset \mathbb{R}$, dense
 in $(-\infty,-1/4)$, such that, if $\lambda\in \Bc$, paths exhibit a periodic behaviour in terms of implementing activities (and thus can only be used to implement a finite set of activities). The following lemma will be important in specifying the set $\Bc$ and understanding this periodic behaviour. The proof is given in Section~\ref{sec:paths}.
\newcommand{\statelemclosedformula}{Let $\lambda< -1/4$ and $\theta\in (0,\pi/2)$ be such that $\lambda= -1/(2\cos\theta)^2$. Then, the partition function of the path $P_n$ with $n$ vertices is given by
\begin{equation*}
Z_{P_n}(\lambda)=\frac{\sin((n+2)\theta)}{2^{n} (\cos\theta)^n\sin(2\theta)}.
\end{equation*}
}
\begin{lemma}\label{lem:closedformula}
\statelemclosedformula
\end{lemma}
 The ``bad'' set $\Bc$ of activities (for which paths exhibit a periodic behaviour) can be read off from Lemma~\ref{lem:closedformula}. To make this precise, let
\begin{equation}\label{eq:bad}
\Bc:=\Big\{\lambda \in \mathbb{R}\mid \lambda=-\frac{1}{4(\cos\theta)^2}\mbox{ for some }\theta \in(0,\pi/2)\mbox{ which is a \emph{rational multiple} of }\pi\Big\}.
\end{equation}
Note, for example, that $-1,-1/2,-1/3\in \Bc$ (set $\theta=\pi/3,\pi/4,\pi/6$, respectively). For $\lambda<-1/4$, it is not hard to infer from Lemma~\ref{lem:closedformula} that  the ratio $\frac{\Zin_{P_n,v}(\lambda)}{\Zout_{P_n,v}(\lambda)}$ is equal to $-\frac{1}{2\cos \theta}\frac{\sin(n\theta)}{\sin((n+1)\theta)}$ (cf. the upcoming equation \eqref{eq:ZNWN}). Therefore, when $\lambda\in\Bc$ or equivalently $\theta$ is a rational multiple of $\pi$, the ratio is periodic in terms of the number of vertices $n$ in the path.  On the other hand, when $\lambda<-1/4$ and $\lambda\notin \Bc$, then we can show that the ratio is  dense in $\mathbb{R}$ as $n$ varies (this follows essentially from the fact that $\{n\theta \bmod 2\pi\mid n\in \mathbb{Z}\}$ is dense on the circle when $\theta$ is irrational) and hence we can use paths to implement a dense set of activities. This is the scope of the next lemma, which is proved in Section~\ref{sec:paths}.
\newcommand{\statelempath}{Let $\lambda<-1/4$ be such that $\lambda\notin \Bc$. Let $P_n$ denote a path with $n$ vertices and let $v$ be one of the endpoints of $P_n$. Then, for every $\lambda'\in \mathbb{R}$, for every $\epsilon>0$, there exists $n$ such that
\begin{equation*}
\Big|\frac{\Zin_{P_n,v}(\lambda)}{\Zout_{P_n,v}(\lambda)}-\lambda'\Big|\leq \epsilon.
\end{equation*}
}
\begin{lemma}\label{lem:path}
\statelempath
\end{lemma}
When $\lambda\in\Bc$, we can no longer use paths to implement a dense set of activities, as we explained earlier, and we need to use a more elaborate argument. A key observation is that, for $\lambda\in \Bc$,  the partition function of a path of appropriate length is equal to 0. In particular, we have the following simple corollary of Lemma~\ref{lem:closedformula}.
\begin{corollary}\label{lem:bad}
Let $\lambda<-1/4$ be such that $\lambda\in \Bc$. Denote by $P_n$ the path with $n$ vertices. Then, there is an integer $n\geq 1$ such that the partition function of the path $P_{n}$ is zero, i.e., $Z_{P_{n}}(\lambda)=0$.
\end{corollary}
\begin{proof}
Since $\lambda \in\Bc$, there exists $\theta\in(0,\pi/2)$ which is a rational multiple of $\pi$ such that $\lambda=-1/(2\cos\theta)^2$. Write $\theta=\frac{p}{q}\pi$ for positive integers $p,q$ satisfying $\gcd(p,q)=1$. Note that $q\geq 3$ since $\theta\in(0,\pi/2)$. By Lemma~\ref{lem:closedformula}, we have that $Z_{P_{q-2}}(\lambda)=0$, as wanted.
\end{proof}
Having a path $P$ whose partition function equals 0 allows us to implement the activity $-1$: indeed, for an endpoint $v$ of the path $P$, we have that
$\Zin_{P,v}(\lambda)+\Zout_{P,v}(\lambda)=Z_P(\lambda)=0$,
and hence $P$, with terminal $v$, implements $\frac{\Zin_{P,v}(\lambda)}{\Zout_{P,v}(\lambda)}=-1$ (note, we will later ensure that $P$ is such that $\Zout_{P,v}(\lambda)\neq 0$). 
Recall that $-1\in \Bc$, so the theorem doesn't follow from Lemma~\ref{lem:path}.
Instead,
after
generalising the partition function to allow non-uniform activities, 
and then using the given activity~$\lambda$, and also the activity~$-1$ that we have already implemented,
a
 somewhat ad-hoc gadget (see Figure~\ref{fig:activityp1}) allows us to also implement the activity $+1$. Using these two implemented activities, $-1$ and $+1$, we then show how to implement \emph{all} rational numbers using graphs whose structure resembles a caterpillar (the proof is inspired by the ``ping-pong'' lemma in group theory, used to establish free subgroups).  We carry out this scheme in a more general setting where, instead of a path, we have a tree whose partition function is zero (this will also be relevant in the regime $\lambda>-1/4$). More precisely, we have the following lemma.
\newcommand{\statelemzero}{Suppose that $\lambda\in \Rnz$ and that $T$ is a tree with $Z_{T}(\lambda)=0$. Let $d$ be the maximum degree of $T$ and let $\Delta=\max\{d,3\}$. Then, $(\Delta,\lambda)$ implements a dense set of activities in $\mathbb{R}$.}
\begin{lemma}\label{lem:zero}
\statelemzero
\end{lemma}
The proof of Lemma~\ref{lem:zero} is given in Section~\ref{sec:zero}. It is immediate to combine Lemma~\ref{lem:path}, Corollary~\ref{lem:bad}, and Lemma~\ref{lem:zero} to obtain the following.
\begin{lemma}\label{lem:largelambda}
Let $\lambda<-1/4$. Then, for $\Delta=3$, $(\Delta,\lambda)$ implements a dense set of activities in $\mathbb{R}$.
\end{lemma}
\begin{proof}
We may assume that $\lambda\in \Bc$, otherwise the result follows directly from Lemma~\ref{lem:path}. For $\lambda\in \Bc$, we have by Corollary~\ref{lem:bad} a path $P$ such that $Z_P(\lambda)=0$. Since $P$ has maximum degree 2, applying Lemma~\ref{lem:zero} gives the desired conclusion.
\end{proof}
Note that Lemma~\ref{lem:largelambda} applies only for values of $\lambda$ which are far from  the threshold $-\lambda^*(\Delta)$ for any $\Delta\geq 3$ and thus it should not be surprising that we can implement a dense set of activities using graphs of maximum degree 3.  This highlights the next obstacle that we have to address: for general degree bounds $\Delta\geq 3$, to get all the way to the threshold $-\lambda^*(\Delta)$ we need to use graphs with maximum degree $\Delta$ (rather than just 3) to have some chance of  implementing 
a dense set of
activities.

Analyzing more complicated graphs for $\Delta\geq 3$ and $-1/4\leq \lambda<-\lambda^*(\Delta)$  might sound daunting given the story for $\lambda<-1/4$, but it turns out that all we need to do is construct a (bipartite) graph $G$ of maximum degree $\Delta$ that implements an activity $\lambda'<-1/4$. Then, to show that $(\Delta,\lambda)$ implements a dense set of activities, we only need to consider whether $\lambda'\in\Bc$. If $\lambda'\notin\Bc$, we can argue by decorating the paths from Lemma~\ref{lem:path} using the graph $G$. Otherwise, if $\lambda'\in\Bc$, we can first construct a tree $T$ of maximum degree $\Delta$ such that $Z_T(\lambda)=0$ (by decorating the path from Corollary~\ref{lem:bad}), and then invoke Lemma~\ref{lem:zero}.  Thus, we are left with the task of implementing  an activity $\lambda'<-1/4$. For that, we combine appropriately $(\Delta-1)$-ary trees of appropriate depth, which can be analysed relatively simply using a recursion. (A technical detail here is that, initially, we are not able to implement this boosted activity $\lambda'$ in the sense of Definition~\ref{def:implement} since the terminal of the relevant tree has degree bigger than 1; nevertheless, the degree of the terminal is at most~$\Delta-2$, so  it can be combined with the paths without overshooting the degree bound $\Delta$.)

Putting together these pieces yields the following lemma, which is proved in Section~\ref{sec:darytree}.
\newcommand{\statelemsmalllambda}{Let $\Delta\geq 3$ and $-1/4\leq \lambda<-\lambda^*(\Delta)$. Then,  $(\Delta,\lambda)$ implements a dense set of activities in $\mathbb{R}$.
}
\begin{lemma}\label{lem:smalllambda}
\statelemsmalllambda
\end{lemma}

Using Lemmas~\ref{lem:largelambda} and~\ref{lem:smalllambda}, the proof of 
our main theorem, Theorem~\ref{thm:precision}, is now immediate. We restate the theorem here for convenience.
\begin{thmprecision}
\statethmprecision
\end{thmprecision}
\begin{proof}
If $\lambda<-1/4$, the lemma follows by Lemma~\ref{lem:largelambda}. If $-1/4\leq \lambda<-\lambda^*(\Delta)$, the lemma follows by Lemma~\ref{lem:smalllambda}.
\end{proof}

In the remainder of this section, we give the proofs of Lemmas~\ref{lem:closedformula},~\ref{lem:path}, \ref{lem:zero}, and  \ref{lem:smalllambda}.

\subsection{Using paths for implementing activities---Proof of Lemmas~\ref{lem:closedformula} and~\ref{lem:path}}\label{sec:paths}
In this section, we prove Lemmas~\ref{lem:closedformula} and~\ref{lem:path}. We start with Lemma~\ref{lem:closedformula}, which we restate here for convenience.
\begin{lemclosedformula}
\statelemclosedformula
\end{lemclosedformula}
\begin{proof}
For $n=-1,0,\hdots$ consider the sequence $r_n$ given by
\[r_n:=\frac{\sin((n+2)\theta)}{2^{n} (\cos\theta)^n\sin(2\theta)}.\]
Also, for $n\geq 1$, let $x_n=Z_{P_{n}}(\lambda)$ denote the partition function of the path $P_n$. It will be useful to extend the sequence $x_n$ for $n=-1,0$ by setting $x_{-1}=x_0=1$. We will show that
\begin{equation}\label{eq:triggoal}
\mbox{for integer $n\geq -1$ it holds that $x_n=r_n$,}
\end{equation}
which clearly yields the lemma (by restricting to $n\geq 1$).

To prove \eqref{eq:triggoal}, note that by considering whether the start of the path belongs to an independent set of $P_n$, we have that for all $n\geq 3$ it holds that
\begin{equation}\label{eq:recnegval}
Z_{P_{n}}(\lambda)=Z_{P_{n-1}}(\lambda)+\lambda Z_{P_{n-2}}(\lambda), \mbox{ or equivalently } x_n=x_{n-1}+\lambda x_{n-2}.
\end{equation}
In fact, since $Z_{P_{1}}(\lambda)=1+\lambda$ and $Z_{P_{2}}(\lambda)=1+2\lambda$, our choice of $x_{-1}=1$ and $x_0=1$ ensures that the second equality in \eqref{eq:recnegval} holds for all integer $n\geq 1$. Using the identity $\sin(2\theta)=2\sin\theta \cos\theta$, we see that $r_{-1}=r_0=1$, so \eqref{eq:triggoal} will follow by showing that for all $n\geq 1$, it holds that
\begin{equation}\label{eq:trig3435}
r_n=r_{n-1}+\lambda r_{n-2}=r_{n-1}-\frac{1}{4(\cos\theta)^2} r_{n-2}.
\end{equation}
The  trigonometric identity $\sin(x+y)+\sin(x-y)=2\sin x\cos y$ gives for $x=(n+1) \theta$, $y=\theta$ that
\begin{equation*}
\sin((n+2)\theta)=2\sin((n+1) \theta)\cos \theta-\sin(n\theta).
\end{equation*}
By dividing this with $2^{n} (\cos\theta)^n \sin(2\theta)\neq 0$, we obtain \eqref{eq:trig3435}, as wanted.
\end{proof}

To prove Lemma~\ref{lem:path}, we will use the following couple of technical lemmas.
\begin{lemma}\label{lem:existence}
Let $a_1,b_1,a_2,b_2\in\mathbb{R}$ be  such that $a_1 b_2 \neq a_2 b_1$ and $a_2^2+b_2^2=1$. Then, for every $x\in \mathbb{R}$, there exist $u,t\in \mathbb{R}$ such that
\[a_2 u+b_2 t\neq 0, \quad x=\frac{a_1 u+b_1 t}{a_2 u+b_2 t},\quad  u^2+t^2=1.\]
\end{lemma}
\begin{proof}
First, note that from  $a_1 b_2 \neq a_2 b_1$ we have that
\begin{gather}
\mbox{ for every $x'\in\mathbb{R}$, at least one of $a_1-a_2x'\neq 0$, $b_1-b_2 x'\neq 0$ holds},\label{eq:projection}\\
\mbox{ for every $u',t'\in\mathbb{R}$ with $u'^2+t'^2=1$, at least one of $a_1u'+b_1t'\neq 0$, $a_2 u'+b_2 t'\neq 0$ holds}.\label{eq:projectionb}
\end{gather}
To see \eqref{eq:projection}, assume for the sake of contradiction that for some $x'\in \mathbb{R}$ we had  $a_1-a_2x'=b_1-b_2 x'= 0$. Then, we would have that $a_1b_2=(a_2x')b_2=a_2( b_2x')=a_2b_1$,
contradicting $a_1 b_2 \neq a_2 b_1$. To see \eqref{eq:projectionb}, for the sake of contradiction, assume that for some $u',t'\in\mathbb{R}$ with 
$u'^2+t'^2=1$ 
we had $a_1u'+b_1t'=a_2 u'+b_2 t'= 0$. Since $u'^2+t'^2=1$, we may assume w.l.o.g. that $u'\neq 0$. Then, we would have that $a_1b_2=(-b_1t'/u')b_2=(-b_2t'/u')b_1=a_2b_1$, contradicting again $a_1 b_2 \neq a_2 b_1$.

Let $x\in \mathbb{R}$ be arbitrary. By applying \eqref{eq:projection} to $x'=x$, we may assume w.l.o.g. that $a_1-a_2x\neq 0$.\footnote{If $a_1-a_2x= 0$, then $b_1-b_2x\neq 0$ and in the subsequent argument one would take $u,t$ such that $t=-\frac{a_1 - a_2 x}{b_1  -b_2  x}u$ and $u^2+t^2=1$.}   Let $u,t$ be such that $u=-\frac{b_1  -b_2  x}{a_1 - a_2 x}t$ and $u^2+t^2=1$. Such $u,t$ clearly exist since, in the $(u,t)$-plane, the first equality is a line through the origin and the second is the unit circle. We claim that $u,t$ satisfy the statement of the lemma. Indeed, the equality $u=-\frac{b_1  -b_2  x}{a_1 - a_2 x}t$ gives
\begin{equation}\label{eq:cdsxza43}
x(a_2u+b_2t)=a_1 u+b_1 t.
\end{equation}
By \eqref{eq:projectionb}, at least one of $a_1u+b_1t\neq 0, a_2u+b_2t\neq 0$ holds, so  from \eqref{eq:cdsxza43}, we obtain that $a_2u+b_2t\neq 0$, yielding also that $x=\frac{a_1 u+b_1 t}{a_2 u+b_2 t}$, as claimed.
\end{proof}

\begin{lemma}\label{lem:crudeapprox}
Let $\theta_1,\theta_2\in \mathbb{R}$ be such that $\sin\theta_2\neq 0$. There exist positive constants $\eta, M$ such that for all $\phi$ with $|\phi|\leq \eta$, it holds that $\sin(\theta_2+\phi)\neq 0$ and
\[\bigg|\frac{\sin(\theta_1+\phi\big)}{\sin\big(\theta_2+\phi\big)}-\frac{\sin\theta_1}{\sin\theta_2}\bigg|\leq M|\phi|.\]
\end{lemma}
\begin{proof}
Let $\eta>0$ be sufficiently small so that $\sin(\theta_2+x)\neq 0$ for all $x\in[-\eta,\eta]$ (such an $\eta$ exists since $\sin \theta_2\neq 0$). Consider the function $f(x)=\frac{\sin(\theta_1+x)}{\sin(\theta_2+x)}$ defined on the interval $I:=[-\eta,\eta]$. By the choice of $\eta$, the function $f$ is well-defined throughout the interval $I$ and has continuous derivative. Set $M:=\max_{x\in I}|f'(x)|$. Then, for all $\phi\in I$, we have  by the Mean Value theorem  that $|f(\phi)-f(0)|\leq M|\phi|$, as wanted.
\end{proof}

We are now ready to prove Lemma~\ref{lem:path}, which we restate here for convenience.
\begin{lempath}
\statelempath
\end{lempath}
\begin{proof}[Proof of Lemma~\ref{lem:path}]
Let $\theta\in (0,\pi/2)$ be such that $\lambda=-1/(2\cos \theta)^2$; since $\lambda<-1/4$, $\theta$ exists and it is unique in the interval $(0,\pi/2)$. From $\lambda\notin \Bc$, we have that $\theta$ is \emph{not} a rational multiple of $\pi$.

For all integer $n\geq 3$, it holds that
\[\Zin_{P_n,v}(\lambda)=\lambda  Z_{P_{n-2}}(\lambda),\qquad  \Zout_{P_n,v}(\lambda)=Z_{P_{n-1}}(\lambda).\]
Since $\theta$ is not a rational multiple of $\pi$, by Lemma~\ref{lem:closedformula} we have that $\Zout_{P_n,v}(\lambda)\neq 0$. Hence, for all $n\geq 3$, we have
\begin{equation}\label{eq:ZNWN}
\frac{\Zin_{P_n,v}(\lambda)}{\Zout_{P_n,v}(\lambda)}=-\frac{1}{2\cos \theta}W_n, \mbox{ where }W_n:=\frac{\sin(n\theta)}{\sin((n+1)\theta)}.
\end{equation}
In fact, it is not hard to see by explicit algebra that \eqref{eq:ZNWN} is valid for $n=1,2$ as well (this can also be inferred from the proof of Lemma~\ref{lem:closedformula}). Thus, to prove the lemma, we only need to show that, for all $w\in \Rnz$, for all $\epsilon>0$, there exists a positive integer $n$ such that $|W_n-w|\leq \epsilon$ (the following argument can also account for $w=0$, by modifying the choice of $\epsilon'$ below, but $w\in \Rnz$ is already sufficient for the density argument and hence we do not need to explicitly do so).

Consider an arbitrary $w\in \Rnz$ and $\epsilon>0$. By Lemma~\ref{lem:existence} applied to $a_1=1,b_1=0,a_2=\cos \theta, b_2=\sin \theta$ and $x=w$, there exist $u,t$ such that $u^2+t^2=1$ and
\[u\cos\theta+t\sin\theta\neq0, \quad w=\frac{u}{u\cos\theta+t\sin\theta}.\]
Since $u^2+t^2=1$ and $w\neq 0$, there is a unique $\theta^*\in (0,2\pi)$ such that $u=\sin\theta^*$ and $t=\cos\theta^*$. For later use, note that with this parametrisation of $u,t$ we have $u\cos\theta+t\sin\theta=\sin(\theta+\theta^*)$ and hence
\[\sin(\theta+\theta^*)\neq 0, \quad w=\frac{\sin\theta^*}{\sin(\theta+\theta^*)}.\]

Let $\eta,M>0$ be the constants in Lemma~\ref{lem:crudeapprox} obtained by setting $\theta_1=\theta^*$ and $\theta_2=\theta+\theta^*$. In particular, we have that for all $\phi$ satisfying $|\phi|\leq \eta$, it holds that $\sin(\theta+\theta^*+\phi)\neq 0$ and
\begin{equation}\label{eq:crudeapprox}
\bigg|\frac{\sin(\theta^*+\phi)}{\sin(\theta+\theta^*+\phi)}-\frac{\sin\theta^*}{\sin(\theta+\theta^*)}\bigg|\leq M |\phi|.
\end{equation}
Let $\epsilon':=\min\{\eta,\epsilon/M, \theta^*,2\pi-\theta^*, 5\theta/\pi\}>0$. Using the fact that $\theta$ is not a rational multiple of $\pi$, we will show that there exist integers $n\geq 1,m\geq 0$ such that
\begin{equation}\label{eq:dirichlet}
|n\theta-(2\pi m+\theta^*)|\leq \epsilon'.
\end{equation}
Before proving this, let us first conclude that for $n$ as in \eqref{eq:dirichlet}, it holds that $|W_n-w|\leq \epsilon$.

Set $\phi:=n\theta-(2\pi m+\theta^*)$ and observe that
\begin{equation*}
\sin(n\theta)=\sin(\theta^*+\phi), \quad \sin((n+1)\theta)=\sin(\theta+\theta^*+\phi).
\end{equation*}
From \eqref{eq:dirichlet} and the choice of $\epsilon'$, we have that $|\phi|\leq  \eta$ and $|\phi|\leq \epsilon/M$, so using \eqref{eq:crudeapprox} we obtain  that
\[|W_n-w|=\bigg|\frac{\sin(n\theta)}{\sin((n+1)\theta)}-\frac{\sin\theta^*}{\sin(\theta+\theta^*)}\bigg|=\bigg|\frac{\sin(\theta^*+\phi)}{\sin(\theta+\theta^*+\phi)}-\frac{\sin\theta^*}{\sin(\theta+\theta^*)}\bigg|\leq M |\phi|\leq \epsilon,\]
as wanted.

It remains to prove \eqref{eq:dirichlet}. By Dirichlet's approximation theorem, there exist positive integers $n',m'$ such that $n' \leq \left\lceil 10/\epsilon'\right\rceil$ and $\big|n'\frac{\theta}{2\pi}-m'|\leq \epsilon'/10$. We thus have that
\[\big|n'\theta-2m'\pi |\leq \epsilon'.\]
Let $z:=n'\theta-2m'\pi$ and note that $z\neq 0$ since $\theta$ is not a rational multiple of $\pi$. We consider two cases depending on the sign of $z$.

\begin{enumerate}[label=\textbf{Case \arabic*}.,  leftmargin=*]
\item $z>0$. By the choice of $\epsilon'$, we have that $\epsilon'\leq \theta^*$, and hence $z\leq \theta^*$. Consider the positive integer
\[k:=\lfloor \theta^*/z\rfloor=\Big\lfloor  \frac{\theta^*}{n'\theta-2m'\pi}\Big\rfloor.\] We claim that \eqref{eq:dirichlet} holds with $n=kn'$ and $m=km'$. Indeed, by the definition of $k,n,m$ we have that
\[ n\theta-2m\pi\leq \theta^*<n\theta-2m\pi+z, \mbox{ so that } |n\theta-(2m\pi+\theta^*)|\leq |z|\leq \epsilon'.\]
\item $z<0$. By the choice of $\epsilon'$, we have that $\epsilon'\leq 2\pi-\theta^*$, and hence $0<-z\leq 2\pi -\theta^*$. Consider the positive integer
\[k:=\lfloor (\theta^*-2\pi)/z\rfloor=\Big\lfloor \frac{2\pi-\theta^*}{2m'\pi-n'\theta}\Big\rfloor.\]
We claim that \eqref{eq:dirichlet} holds with $n=kn'$ and $m=km'-1$.
Indeed, by the definition of $k,n,m$ we have that
\[ z< 2m\pi+\theta^*-n\theta\leq 0, \mbox{ so that } |n\theta-(2m\pi+\theta^*)|\leq |z|\leq \epsilon'.\]
\end{enumerate}
This concludes the proof of \eqref{eq:dirichlet} and hence the proof of Lemma~\ref{lem:path}.
\end{proof}

\subsection{The ratio $R_\lambda(G,v)$ and a simple way to implement activities}\label{sec:Rl}
Let $\lambda\in \Rnz$. To prove Lemmas \ref{lem:zero} and  \ref{lem:smalllambda}, it will be sometimes more convenient to work with the ratio $\frac{\Zout_{G,v}(\lambda)}{Z_{G}(\lambda)}$ (rather than $\frac{\Zin_{G,v}(\lambda)}{\Zout_{G,v}(\lambda)}$). Formally, let $G=(V,E)$ be a graph such that $Z_G(\lambda)\neq 0$. Then, for a vertex $v\in V$, we will be interested in the quantity $R_\lambda(G,v)$ defined as
\[R_\lambda(G,v):=\frac{\Zout_{G,v}(\lambda)}{Z_G(\lambda)}.\]

The following simple lemma shows how to implement activities using the quantities $R_\lambda(G,v)$.
\begin{lemma}\label{lem:translation}
Let $\lambda\in \Rnz$ and $r\in \mathbb{R}$. Let $G'$ be a graph and $u$ be a vertex of $G'$ such that $r=R_\lambda(G',u)$. Consider the graph $G$ obtained from $G'$ by adding a new vertex $v$ whose single neighbour is the vertex $u$. Then,
\[\Zout_{G,v}(\lambda)\neq 0\mbox{ and }\frac{\Zin_{G,v}(\lambda)}{\Zout_{G,v}(\lambda)}=\lambda r.\]
\end{lemma}
\begin{proof}
Note that $\Zout_{G,v}(\lambda)=Z_{G'}(\lambda)\neq 0$, where the disequality follows from the fact that the quantity $R_\lambda(G',u)$ is well-defined. Then, observe that
\[\frac{\Zin_{G,v}(\lambda)}{\Zout_{G,v}(\lambda)}=\frac{\lambda\Zout_{G',u}(\lambda)}{Z_{G'}(\lambda)}=\lambda\cdot R_\lambda(G',u)=\lambda r.\qedhere\]
\end{proof}
It is instructive at this point to note the following consequence of Lemma~\ref{lem:translation}: to show that $(\Delta,\lambda)$ implements an activity $\lambda'$, it suffices to construct a bipartite graph $G$ with maximum degree at most $\Delta$ which has a vertex $v$ whose degree is at most $\Delta-1$ such that $R_\lambda(G,v)=\lambda'/\lambda$.

\subsection{The case where the partition function of some tree is zero --- Proof of Lemma~\ref{lem:zero}.}\label{sec:zero}
In this section, we prove Lemma~\ref{lem:zero}. 
We start by defining a multivariate version of the hard-core partition function.

\subsubsection{The hard-core model with non-uniform activities}\label{sec:nonuniform}

Implementing activities can be thought of as constructing unary gadgets that allow  modification of the activity at a particular vertex $v$. We will use the implemented activities to simulate a more general version of the hard-core model with non-uniform activities. In particular, let $G=(V,E)$ be a graph and $\lambdab=\{\lambda_v\}_{v\in V}$ be a real vector; we associate to every vertex $v\in V$ the activity $\lambda_v$. The hard-core partition function with activity vector $\lambdab$ is defined as
\[Z_G(\lambdab):=\sum_{I\in \mathcal{I}_G}\prod_{v\in I} \lambda_v.\]
Note that the standard hard-core model with activity $\lambda$ is obtained from this general version by setting all vertex activities equal to $\lambda$. For a vertex $v\in V$, we define $\Zin_G(\lambdab)$ and $\Zout_G(\lambdab)$ for the non-uniform model analogously to $\Zin_G(\lambda)$ and $\Zout_G(\lambda)$ for the uniform model, respectively.

The following lemma connects the partition function $Z_G(\lambdab)$ with non-uniform activities to the hard-core partition function with uniform activity $\lambda$. Roughly, whenever all the activities in the activity vector $\lambdab$ can be implemented, we can just stick graphs on the vertices of $G$ which implement the corresponding activities in $\lambdab$ (if a vertex activity equals $\lambda$, no action is required).  
\begin{lemma}\label{lem:nonuniform}
Let $\lambda\in \Rnz$, let $t\geq 1$ be an arbitrary integer and let $\lambda_1',\hdots,\lambda_t'\in \mathbb{R}$. Suppose that, for $j\in [t]$, the graph $G_j$ with terminal $v_j$ implements the activity $\lambda_j'$, and let  $C_j:=\Zout_{G_j,v_j}(\lambda)$. Then, the following holds for every graph $G=(V,E)$ and every activity vector $\lambdab=\{\lambda_v\}_{v\in V}$ such that $\lambda_v\in\{\lambda,\lambda_1',\hdots,\lambda_t'\}$ for every $v\in V$.

For $j\in[t]$, let $V_j:=\{v\in V\mid \lambda_v=\lambda_j'\}$. Consider the graph $G'$ obtained from $G$ by attaching, for every $j\in [t]$ and every vertex $v\in V_j$, a copy of the graph $G_j$ to the vertex $v$ and identifying the terminal $v_j$ with the vertex $v$ (see Figure~\ref{fig:zaq}). Then, for $C:=\prod^t_{j=1}C_j^{|V_j|}$, it holds that
\begin{gather}
Z_{G'}(\lambda)=C\cdot Z_G(\lambdab),\label{eq:uni}\\
\mbox{$\forall v\in V$: } \Zin_{G',v}(\lambda)=C\cdot \Zin_{G,v}(\lambdab), \qquad \Zout_{G',v}(\lambda)=C\cdot \Zout_{G,v}(\lambdab).\label{eq:uniinout}
\end{gather}
\end{lemma}
\begin{remark}\label{rem:blowup}
Note that, in the construction of Lemma~\ref{lem:nonuniform}, every vertex $v\in G$ with $\lambda_v=\lambda$ maintains its degree in $G'$ (in fact, the neighbourhood of such a vertex $v$ is the same in $G$ and $G'$). The degree of every other vertex $v$ in $G$ gets increased by one. This observation will ensure in applications of Lemma~\ref{lem:nonuniform} that we do not blow up the degree.

Note also that, if the graph $G$ is bipartite and the graphs $G_j$ are bipartite for all $j=1,\hdots,t$, then the graph $G'$ in the construction of Lemma~\ref{lem:nonuniform} is bipartite as well.  This observation will ensure in applications of Lemma~\ref{lem:nonuniform} that we preserve the bipartiteness of the underlying graph $G$.
\end{remark}

\begin{figure}[h]
\captionsetup[subfigure]{justification=centering}
\begin{minipage}{0.4\textwidth}
\begin{center}
{
\tikzset{lab/.style={circle,draw,inner sep=0pt,fill=none,minimum size=5mm}}
\begin{tikzpicture}[xscale=1,yscale=1]
\draw (0,0) node[lab] (1) {$v_1$};
\draw (1,0) node[lab] (2) {};
\draw (2,0) node[lab] (3) {};
\draw (1) -- (2) -- (3);
\end{tikzpicture}
}
\end{center}
\subcaption{The graph $G_1$ with terminal $v_1$ implementing an activity $\lambda_1'$}
\label{fig:G1v1}
\par\bigskip

\begin{center}
{
\tikzset{lab/.style={circle,draw,inner sep=0pt,fill=none,minimum size=5mm}}
\begin{tikzpicture}[xscale=1,yscale=1]
\draw (0,0) node[lab] (1) {$v_2$};
\draw (1,0) node[lab] (2) {};
\draw (1.866,1/2) node[lab] (3) {};
\draw (1.866,-1/2) node[lab] (4) {};
\draw (1) -- (2) -- (3);
\draw (2) -- (4);
\end{tikzpicture}
}
\end{center}
\subcaption{The graph $G_2$ with terminal $v_2$ implementing an activity $\lambda_2'$}
\label{fig:G2v2}
\end{minipage}
\begin{minipage}{0.6\textwidth}
\begin{center}
{
\tikzset{lab/.style={circle,draw,inner sep=0pt,fill=none,minimum size=5mm}}
\begin{tikzpicture}[xscale=1,yscale=1]
\draw (0,0) node[lab, label={below:$\lambda_1=\lambda_1'$}] (1) {$1$};
\draw (2,0) node[lab, label={below:$\lambda_2=\lambda_2'$}] (2) {$2$};
\draw (4,0) node[lab, label={below:$\lambda_3=\lambda$}] (3) {$3$};
\draw (6,0) node[lab, label={below:$\lambda_4=\lambda_1'$}] (4) {$4$};
\draw (1) -- (2) -- (3) -- (4);
\end{tikzpicture}
}
\end{center}
\subcaption{The graph $G$ with activity vector $\lambdab$}
\label{fig:Glambdab}
\par\medskip

\begin{center}
{
\tikzset{lab/.style={circle,draw,inner sep=0pt,fill=none,minimum size=5mm}}
\begin{tikzpicture}[xscale=1,yscale=1]
\draw (0,0) node[lab] (1) {$1$};
\draw (0,-1) node[lab] (1a) {};
\draw (0,-2) node[lab] (1b) {};
\draw (2,0) node[lab] (2) {$2$};
\draw (2,-1) node[lab] (2a) {};
\draw (2.5,-1.866) node[lab] (2b) {};
\draw (1.5,-1.866) node[lab] (2c) {};
\draw (4,0) node[lab] (3) {$3$};
\draw (6,0) node[lab] (4) {$4$};
\draw (6,-1) node[lab] (4a) {};
\draw (6,-2) node[lab] (4b) {};
\draw (1) -- (2) -- (3) -- (4);
\draw (1) -- (1a) -- (1b);
\draw (4) -- (4a) -- (4b);
\draw (2) -- (2a) -- (2b);
\draw (2a) -- (2c);
\end{tikzpicture}
}
\end{center}
\subcaption{The graph $G'$ with uniform activity $\lambda$}
\label{fig:Gplambda}
\end{minipage}
\caption{An illustrative depiction of the construction in the statement of Lemma~\ref{lem:nonuniform}. The graphs $G_1,G_2$ in Figures \ref{fig:G1v1}, \ref{fig:G2v2} implement the activities $\lambda_1',\lambda_2'$, respectively, i.e., $\frac{\Zin_{G_1,v_1}(\lambda)}{\Zout_{G_1,v_1}(\lambda)}=\lambda_1'$ and $\frac{\Zin_{G_2,v_2}(\lambda)}{\Zout_{G_2,v_2}(\lambda)}=\lambda_2'$ for some $\lambda_1',\lambda_2'\in \Reals$. In Figure~\ref{fig:Glambdab}, we have a graph $G$ with non-uniform activities $\{\lambda_i\}_{i\in[4]}$ such that $\lambda_i\in \{\lambda,\lambda_1',\lambda_2'\}$ for $i\in [4]$. By sticking onto $G$ the graphs $G_1,G_2$ as in Figure~\ref{fig:Gplambda}, we obtain the graph $G'$. Note that the vertex whose activity was equal to $\lambda$ was not modified.}
\label{fig:zaq}
\end{figure}
\begin{proof}[Proof of Lemma~\ref{lem:nonuniform}]
For an independent set $I$ of $G$, let
$\Omega_I=\{I'\in \mathcal{I}_{G'}\mid I'\cap V=I\}$,
i.e., $\Omega_I$ is the set of independent sets in $G'$ whose restriction on $G$ is the independent set $I$. We will show that
\begin{equation}\label{eq:basicI}
\sum_{I'\in \Omega_I} \lambda^{|I'|}=C\cdot \prod_{v\in I} \lambda_v,
\end{equation}
where $C=\prod^t_{j=1}C_j^{|V_j|}$ is as in the statement of the lemma. Note that the sets $\{\Omega_{I}\}_{I\in \mathcal{I}_G}$ are a partition of the set $\mathcal{I}_{G'}$ of independent sets of $G'$, so summing \eqref{eq:basicI} over $I\in \mathcal{I}_G$ yields \eqref{eq:uni}. Analogously, for any vertex $v$ of the graph $G$,  we may sum \eqref{eq:basicI} over those $I\in \mathcal{I}_G$ such that $v\in I$ to obtain the first equality in \eqref{eq:uniinout}; by summing over those  $I\in \mathcal{I}_G$ such that $v\notin I$, we also obtain the second equality in \eqref{eq:uniinout}. We thus focus on proving \eqref{eq:basicI}.

Let $V_0=V\backslash (V_1\cup \cdots V_t)$, i.e., $V_0$ consists of all the vertices in $G$ such that $\lambda_v=\lambda$. Further, for each $j\in [t]$ and $v\in V_j$, denote by $G^v_j$ the copy of the graph $G_j$ which is attached to $v$. Consider the product
\[P:=\lambda^{|V_0\cap I|}\prod^t_{j=1}\prod_{v\in V_j\cap I} \Zin_{G^v_j,v}(\lambda)\prod_{v\in V_j\backslash I} \Zout_{G^v_j,v}(\lambda).\]
We claim that $P=\sum_{I'\in \Omega_I} \lambda^{|I'|}$. Indeed, using the definition of $\Zin_{G^v_j,v}(\lambda)$ and $\Zout_{G^v_j,v}(\lambda)$, we can rewrite $P$ as
\[P=\lambda^{|V_0\cap I|}\prod^{t}_{j=1}\prod_{v\in V_j\cap I}\bigg(\sum_{I'\in \mathcal{I}_{G^v_j}; v\in I'}\lambda^{|I'|}\bigg)\prod_{v\in V_j\backslash I}\bigg(\sum_{I'\in \mathcal{I}_{G^v_j}; v\notin I'}\lambda^{|I'|}\bigg).\]
By multiplying out the last expression, we recover precisely the sum $\sum_{I'\in \Omega_I} \lambda^{|I'|}$.

To prove \eqref{eq:basicI}, it thus remains to massage $P$ into the r.h.s. of \eqref{eq:basicI}. Let $j\in [t]$ and $v\in V_j$. By construction, $G^v_j$ is a copy of $G_j$ and $v$ has been identified with the terminal $v_j$ of $G_j$, so  $\Zin_{G^v_j,v}(\lambda)=\Zin_{G_j,v_j}(\lambda)$ and $\Zout_{G^v_j,v}(\lambda)=\Zout_{G_j,v_j}(\lambda)$. Recall also that $G_j$ implements $\lambda_j'$ and, in particular, $\lambda_v=\lambda_j'=\Zin_{G_j,v_j}(\lambda)/\Zout_{G_j,v_j}(\lambda)$. It follows that for every $j\in[t]$ it holds that
\[\prod_{v\in V_j\cap I} \Zin_{G^v_j,v}(\lambda)\prod_{v\in V_j\backslash I} \Zout_{G^v_j,v}(\lambda)=\prod_{v\in V_j\cap I} \Zin_{G_j,v_j}(\lambda)\prod_{v\in V_j\backslash I} \Zout_{G_j,v_j}(\lambda)= C_j^{|V_j|} (\lambda_v)^{|V_j\cap I|},\]
where we recall that $C_j=\Zout_{G_j,v_j}(\lambda)$. By multiplying this over $j\in[t]$, we obtain that $P=C\cdot \prod_{v\in I} \lambda_v$, as wanted. This proves \eqref{eq:basicI} and completes the proof of Lemma~\ref{lem:nonuniform}.
\end{proof}

\subsubsection{The proof of Lemma~\ref{lem:zero}}

Having extended the hard-core model to have non-uniform activities, we proceed with the proof of Lemma~\ref{lem:zero}.
We start with the following lemma.
\begin{lemma}\label{lem:minusone}
Let $\lambda\in \Rnz$ and $d\geq 2$ be a positive integer. Suppose that there is a tree $T$ with maximum degree $d$ such that $Z_T(\lambda)=0$. Then, for $\Delta=\max\{d,3\}$, we have that $(\Delta,\lambda)$ implements the activities $-1$ and $+1$.
\end{lemma}
\begin{proof}
\textbf{Part  I.} We first prove that $(\Delta,\lambda)$ implements the activity $-1$. 

If $\lambda=-1$, we have that the path $P$ with length 3 implements the activity $-1$, since for an endpoint $u$ of $P$ it holds that $\Zout_{P,u}(\lambda)=1+3\lambda+\lambda^2=-1$ and $\Zin_{P,u}(\lambda)=\lambda(1+2\lambda)=1$. (The reason that we use a  path of length 3 rather than a single-vertex path is to ensure that the terminal of the path has degree 1, as Definition~\ref{def:implement} requires. This will be technically convenient in the upcoming proof of Lemma~\ref{lem:implementingS}.) 

We will therefore assume that $\lambda\neq -1$. By assumption, there exists a tree $T$ of maximum degree $d$ satifying $Z_{T}(\lambda)=0$. Among the trees $T$ of maximum degree at most $d$ satisfying $Z_{T}(\lambda)=0$, let $T^*$ be a tree which has the minimum number of vertices. Since the partition function of a single-vertex graph is $1+\lambda\neq 0$, we have that the tree $T^*$ has at least two vertices. We thus conclude that for every leaf $u$ of $T$ it holds that $\Zout_{T^*,u}(\lambda)\neq 0$: if not, the tree $T^*\backslash u$ satisfies $Z_{T^*\backslash u}(\lambda)=\Zout_{T^*,u}(\lambda)=0$, contradicting the minimality of $T^*$.

Let $u$ be an arbitrary leaf of the tree $T^*$. Since
\begin{equation*}
0=Z_{T^*}(\lambda)=\Zin_{T^*,u}(\lambda)+\Zout_{T^*,u}(\lambda)\mbox{ and }\Zout_{T^*,u}(\lambda)\neq 0,
\end{equation*}
we have that the tree $T^*$ (with terminal $u$) implements the activity $\frac{\Zin_{T^*,u}(\lambda)}{\Zout_{T^*,u}(\lambda)}=-1$. This completes the proof that $(\Delta,\lambda)$ implements the activity $-1$.

\textbf{Part II.} We next show that $(\Delta,\lambda)$ implements the activity $+1$.

We will make use of the hard-core model with non-uniform activities, cf. Section~\ref{sec:nonuniform}. In particular, consider the (bipartite) graph $G$ in Figure~\ref{fig:activityp1}, with vertex activities that are also given in the figure.
\begin{figure}[t]
\begin{center}
{
\tikzset{lab/.style={circle,draw,inner sep=0pt,fill=none,minimum size=5mm}}
\begin{tikzpicture}[xscale=1,yscale=1]

\draw (-6,0) node[lab, label={$\lambda_v=\lambda$}] (1) {$v$};
\draw (-4,0) node[lab, label={$\lambda_1=-1$}] (2) {$1$};
\draw (-2,0) node[lab, label={$\lambda_2=-1$}] (3) {$2$};
\draw (0,0) node[lab, label={$\lambda_3=\lambda$}] (4) {$3$};
\draw (1.732,1) node[lab, label={above:$\lambda_4=\lambda$}] (5) {$4$};
\draw (3.732,1) node[lab, label={above:$\lambda_5=-1$}] (6) {$5$};
\draw (1.732,-1) node[lab, label={above:$\lambda_6=-1$}] (7) {$6$};
\draw (3.732, -1) node[lab, label={above:$\lambda_7=\lambda$}] (8) {$7$};
\draw (5.732, 0) node[lab, label={above:$\lambda_8=-1$}] (9) {$8$};
\draw (1) -- (2) -- (3) -- (4) -- (5) -- (6) -- (9);
\draw (4) -- (7) -- (8) -- (9);
\end{tikzpicture}
}
\end{center}
\caption{The bipartite graph $G$ with nonuniform activities $\lambdab$ used in Lemma~\ref{lem:minusone} to prove that we can implement the activity $+1$. We show that $\Zin_{G,v}(\lambdab)=\Zout_{G,v}(\lambdab)=-\lambda^2$, see \eqref{eq:adhocgadget}. By invoking Lemma~\ref{lem:nonuniform}, we obtain a bipartite graph $G'$ with uniform activities equal to $\lambda$ and whose terminal is the vertex $v$ such that $\frac{\Zin_{G',v}(\lambda)}{\Zout_{G',v}(\lambda)}=\frac{\Zin_{G,v}(\lambdab)}{\Zout_{G,v}(\lambdab)}=+1$. It follows that $G'$ implements the activity $+1$.}
\label{fig:activityp1}
\end{figure}
We will use $\lambdab$ to denote the activity vector on $G$. By enumerating the independent sets $I$ of $G$, we compute
\begin{equation}\label{eq:adhocgadget}
\begin{aligned}
\Zin_{G,v}(\lambdab)&=\lambda_v\lambda_3(1+\lambda_5+\lambda_7+\lambda_8+\lambda_5\lambda_7)\\
&\hskip 1cm+\lambda_v(1+\lambda_2)(1+\lambda_4+\lambda_5+\lambda_6+\lambda_7+\lambda_8\\
&\hskip 4cm+\lambda_4\lambda_6+\lambda_4\lambda_7+\lambda_4\lambda_8+\lambda_5\lambda_6+\lambda_5\lambda_7+\lambda_6\lambda_8+\lambda_4\lambda_6\lambda_8)\\
&=-\lambda^2,\\
\Zout_{G,v}(\lambdab)&=\lambda_3(1+\lambda_1)(1+\lambda_5+\lambda_7+\lambda_8+\lambda_5\lambda_7)\\
&\hskip 1cm+(1+\lambda_1+\lambda_2)(1+\lambda_4+\lambda_5+\lambda_6+\lambda_7+\lambda_8\\
&\hskip 4cm+\lambda_4\lambda_6+\lambda_4\lambda_7+\lambda_4\lambda_8+\lambda_5\lambda_6+\lambda_5\lambda_7+\lambda_6\lambda_8+\lambda_4\lambda_6\lambda_8)\\
&=-\lambda^2,\\
\end{aligned}
\end{equation}
Since $(\Delta,\lambda)$ implements the activity $-1$ by the first part of the proof, we may invoke Lemma~\ref{lem:nonuniform} (see also Remark~\ref{rem:blowup}) to obtain a bipartite graph $G'=(V',E')$ of maximum degree $\Delta$ such that $v\in V'$ and
\[\frac{\Zin_{G',v}(\lambda)}{\Zout_{G',v}(\lambda)}=\frac{\Zin_{G,v}(\lambdab)}{\Zout_{G,v}(\lambdab)}=+1.\]
Further, by the construction in Lemma~\ref{lem:nonuniform}, $v$ continues to have degree 1 in $G'$. It follows that $G'$ (with terminal $v$) implements the activity $+1$, as wanted.
This completes   the proof of Lemma~\ref{lem:minusone}.
\end{proof}

The following functions $f_+$ and $f_-$ will be important in what follows:
\begin{equation*}
\begin{gathered}
f_+:\mathbb{R}\backslash\{-1\}\mapsto \mathbb{R}\backslash\{0\},\mbox{ given by }f_+(x)=\frac{1}{1+x}\mbox{ for all } x\neq -1,\\
f_-:\mathbb{R}\backslash\{+1\}\mapsto \mathbb{R}\backslash\{0\},\mbox{ given by }f_-(x)=\frac{1}{1-x}\mbox{ for all } x\neq +1.
\end{gathered}
\end{equation*}
\begin{definition}\label{def:S}
Let $S\subseteq \mathbb{R}$ be the set of real numbers defined as follows: $z\in S$ iff for some integer $n\geq 0$, there exists a sequence $x_0,\hdots,x_n$ such that $x_0=0$, $x_n=z$ and for all $i=0,\hdots,n-1$ it holds that
\[\mbox{ either }x_{i+1}=f_+(x_i) \mbox{ or } x_{i+1}=f_-(x_i).\]
\end{definition}
In other words, the set $S$ in Definition~\ref{def:S} can be obtained by the following recursive procedure. Initialise $S_0=\{0\}$. For $h=0,1,\hdots$, define $S_{h+1}$ by first letting $S^+_{h+1}=f_+(S_{h})$ and  $S^-_{h+1}=f_-(S_h)$  and then setting $S_{h+1}=S^+_{h+1}\cup S^-_{h+1}$. The set $S$  can then be recovered by taking the union of the sets $S_h$, i.e., $S=\cup^{\infty}_{h=0}S_h$.

Our interest in the set $S$ is justified by the following lemma.
\begin{lemma}\label{lem:implementingS}
Let $\Delta\geq 3$ and $\lambda<0$. Suppose that $(\Delta,\lambda)$ implements the activities $-1$ and $+1$. Then, $(\Delta,\lambda)$ also implements the set of activities $\{\lambda z\mid z\in S\}$, where $S\subseteq \mathbb{R}$ is given in Definition~\ref{def:S}.
\end{lemma}
\begin{proof}
For simplicity, we drop the $\lambda$'s from notation, i.e., we will just write $Z_G,\Zout_{G,v},\Zin_{G,v}$ instead of $Z_G(\lambda),\Zout_{G,v}(\lambda),\Zin_{G,v}(\lambda)$.

Since $(\Delta,\lambda)$ implements the activities $-1$ and $+1$, we have that there exist bipartite graphs $G_+,G_-$ of maximum degree at most $\Delta$ with terminals $v_+,v_-$, respectively, such that
\begin{gather*}
\Zout_{G_+,v_+}\neq 0\mbox{ and }\frac{\Zin_{G_+,v_+}}{\Zout_{G_+,v_+}}=+1,\quad \Zout_{G_-,v_-}\neq 0 \mbox{ and } \frac{\Zin_{G_-,v_-}}{\Zout_{G_-,v_-}}=-1.
\end{gather*}
Recall also (from Definitions~\ref{def:Gimplement} and~\ref{def:implement}) that $v_+,v_-$ have degree 1 in $G_+$ and $G_-$, respectively.

Consider an arbitrary $z\in S$. Then, there exists a sequence  $x_0,\hdots,x_n$ such that $x_0=0$, $x_n=z$ and for all $i=0,\hdots,n-1$ it holds that
\[\mbox{ either }x_{i+1}=f_+(x_i) \mbox{ or } x_{i+1}=f_-(x_i).\]
For $i=0,\hdots,n$, we will construct inductively  a bipartite graph $G_i$ of maximum degree $\Delta$ and specify a vertex $u_i$ in $G_i$ with degree at most $2$ such that $R(G_i,u_i)=x_i$. Using this for $i=n$ in conjuction with Lemma~\ref{lem:translation} yields that $(\Delta,\lambda)$ implements $\lambda z$, as wanted.

We begin with the base case $i=0$. Consider the graph $G$ obtained by adding a new vertex $u$ to the graph $G_-$ and connecting $u$ and $v_-$ with an edge. Note that $G$ is bipartite and has maximum degree at most $\Delta$; also,  $u$ has degree 1 in $G$.  Further, we have
\begin{gather*}
\Zout_{G,u}=Z_{G_-}=\Zin_{G_-,v_-}+\Zout_{G_-,v_-}=0,\\
Z_G=\Zin_{G,u}+\Zout_{G,u}=\Zin_{G,u}=\lambda\cdot \Zout_{G_-,v_-}\neq 0.
\end{gather*}
It follows that $R(G,u)=\Zout_{G,u}/Z_G=0$, which completes the proof for $i=0$ by setting $G_0=G$ and $u_0=u$.

For the induction step, let $i$ be an integer satisfying $0\leq i\leq n-1$ and assume that $G_i$ is a bipartite graph of maximum degree at most $\Delta$ such that $R(G_i,u_i)=x_{i}$ for some vertex $u_i$  whose degree is at most two. Let $s\in \{+,-\}$ be such that $x_{i+1}=f_s(x_{i})$.

Let $G$ be the bipartite graph obtained as follows: take a copy of $G_s$, a copy of $G_i$ and connect the vertices $v_s$ and $u_i$ with an edge. Note that $G$ has maximum degree $\Delta$ since the only vertices whose degree has increased (by one) are the vertices $u_i$ and $v_s$; the vertex $v_s$ has degree two in $G$ and $u_i$ has degree at most three.

We next show that $R(G,v_s)=x_{i+1}$ which establishes the induction step by setting $G_{i+1}=G$ and $u_{i+1}=v_s$. We have that
\[\Zout_{G,v_s}=\Zout_{G_s, v_s}Z_{G_{i}}, \quad \Zin_{G,v_s}=\Zin_{G_s,v_s} \Zout_{G_{i},u_i}, \quad Z_{G}=\Zin_{G,v_s}+\Zout_{G,v_s},\]
so that
\[R(G,v_s)=\frac{\Zout_{G,v_s}}{Z_{G}}=\frac{\Zout_{G_s, v_s} Z_{G_{i}}}{\Zout_{G_s, v_s} Z_{G_{i}}+\Zin_{G_s,v_s}\Zout_{G_{i},v_i}}=\frac{1}{1+\frac{\Zin_{G_s,v_s}}{\Zout_{G_s, v_s}}x_i}=f_s(x_i)=x_{i+1}.\]
This concludes the proof of Lemma~\ref{lem:implementingS}.
\end{proof}

It is simple to see that all numbers in the set $S$ of Definition~\ref{def:S} are rationals. Somewhat surprisingly, the following lemma asserts that $S$ is in fact the set $\mathbb{Q}$ of all rational numbers.
\begin{lemma}\label{lem:SequalsQ}
Let $S\subseteq \mathbb{R}$ be the set in Definition~\ref{def:S}. Then, $S=\mathbb{Q}$, i.e., $S$ is the set of all rational numbers. 
\end{lemma}
\begin{proof}
As noted earlier, it is simple to see that $S\subseteq \mathbb{Q}$. Thus, we only need to argue that $\mathbb{Q}\subseteq S$.

Since $0\in S$ (by taking $n=0$ in Definition~\ref{def:S}) and $f_+(0)=1$, we have that $0,1\in S$. Note that
\begin{equation}\label{eq:oppositesign}
f_-(f_-(f_+(x)))=-x \mbox{ for } x\neq -1,0.
\end{equation}
It follows that $-1\in S$.  Observe also that $\frac{1}{2},2\in S$ since $f_+(f_+(0))=1/2$ and $f_-(f_+(f_+(0)))=2$.

Let $T:=\{-1,0,1/2,1,2\}$; the arguments above established that $T\subseteq S$. Consider an arbitrary $\rho \in \mathbb{Q}$ such that $\rho\notin T$. To prove the lemma, we need to show that $\rho\in S$.

We will show that, for some integer $n\geq 0$, there is a sequence $\{\rho_{i}\}^n_{i=0}$  such that
\begin{enumerate}[label=(\roman*),ref=(\roman*),leftmargin=*]
\item \label{it:r0rn} $\rho_0=\rho$, $\rho_n=-1$.
\item \label{it:rineq} $\rho_i\notin \{0,1/2,1\}$ for  $i=0,\hdots,n-1$.
\item \label{it:rieq} $\rho_{i+1}=f_+(\rho_i)$ or  $\rho_{i+1}=f_-(\rho_i)$ for $i=0,\hdots,n-1$.
\end{enumerate}
Before proving the existence of such a sequence, we first show how to conclude that $\rho\in S$. To do this,  let $x_i:=\rho_{n-i}$ for  $i=0,\hdots,n$. Properties \ref{it:r0rn}--\ref{it:rieq} of the sequence $\{\rho_{i}\}^n_{i=0}$ translate into the following properties of the sequence $\{x_{i}\}^n_{i=0}$:
\begin{enumerate}[label=(\alph*),ref=(\alph*),leftmargin=*]
\item \label{it:x0xn} $x_0=-1$, $x_n=\rho$.
\item \label{it:xineq} $x_i\notin \{0,1/2,1\}$ for $i=1,2,\hdots,n$.
\item \label{it:xieq} $x_{i}=f^{-1}_+(x_{i-1})$ or  $x_{i}=f^{-1}_-(x_{i-1})$ for $i=1,2,\hdots,n$.
\end{enumerate}
We show by induction on $i$ that $x_i\in S$ for all $i=0,\hdots,n$, which for $i=n$ gives that $\rho\in S$ (since by Item~\ref{it:x0xn} we have $x_n=\rho$). For the base case $i=0$, we have that $x_0=-1$ by Item~\ref{it:x0xn} and hence $x_0\in S$. For the induction step, assume that $x_i\in S$ for some integer $0\leq i\leq n-1$, our goal is to show that $x_{i+1}\in S$. The main observation is that the inverses of the functions $f_-$ and $f_+$ can be obtained by composing appropriately the functions $f_-$ and $f_+$. Namely, we have that
\begin{gather}
f^{-1}_-(x)=\frac{x-1}{x}=f_-(f_-(x)) \mbox{ for $x\neq 0,1$},\label{eq:fminverse}\\
f^{-1}_+(x)=\frac{1-x}{x}=f_-(f_-(f_+(f_-(f_-(x)))))) \mbox{ for $x\neq 0,\frac{1}{2},1$}.\label{eq:fpinverse}
\end{gather}
\eqref{eq:fminverse} is proved by just making the substitutions. \eqref{eq:fpinverse} is obtained from  \eqref{eq:oppositesign} and \eqref{eq:fminverse}, and checking when $f_-(f_-(x))=\frac{x-1}{x}$ equals $-1$ and $0$. Since by Items~\ref{it:x0xn} and~\ref{it:xineq} we have that $x_j\neq 0,1/2,1$ for all $0\leq j\leq n$ and $x_i\in S$ by the induction hypothesis, it follows by Item~\ref{it:xieq} and \eqref{eq:fminverse}, \eqref{eq:fpinverse} that $x_{i+1}\in S$, as wanted.

It remains to establish the existence of the sequence $\{\rho_i\}^n_{i=0}$ with the properties \ref{it:r0rn}--\ref{it:rieq}. Consider the following set $S_\rho$, which is defined analogously to the set $S$ with the only difference that the starting point for $S_\rho$ is the point $\rho$ (instead of $0$ that was used in the definition of $S$). Formally, $z\in S_\rho$ iff for some integer $n\geq 0$, there exists a sequence $\{\rho_i\}^n_{i=0}$ such that $\rho_0=\rho$, $\rho_n=z$ and for all $i=0,\hdots,n-1$ it holds that either $\rho_{i+1}=f_+(\rho_i)$ or $x_{i+1}=f_-(\rho_i)$. For convenience, we will call such a sequence a certificate that $z\in S_\rho$ and we will refer to $n$ as the length of the certificate.

We will show that, for any $\rho\in\mathbb{Q}$ such that $\rho\notin T=\{-1,0,1/2,1,2\}$, it holds that
\begin{equation}\label{eq:Srho}
0,1\notin S_\rho, \qquad -1\in S_\rho.
\end{equation}
Prior to that, let us use \eqref{eq:Srho} to establish the existence of the desired sequence. Among all certificates that $-1\in S_\rho$,  consider one with the smallest possible length, which we will denote by $\{\rho_i\}^n_{i=0}$. (The existence of such a certificate is guaranteed by \eqref{eq:Srho}.) We claim that the sequence  $\{\rho_i\}^n_{i=0}$ has all of the required properties~\ref{it:r0rn},~\ref{it:rineq}, and \ref{it:rieq}. By the definition of a certificate and since $\rho_n=-1$, we have that the sequence  $\{\rho_i\}^n_{i=0}$ satisfies automatically properties~\ref{it:r0rn} and~\ref{it:rieq}. To prove property~\ref{it:rineq}, we  argue that $\rho_i\not\in \{0,1/2,1\}$ for all integers $0\leq i\leq n$. We  cannot have  an  $i$ such that $\rho_i=0$ or $\rho_i=1$ since this would contradict that $0,1\notin S_\rho$ (by \eqref{eq:Srho}). Suppose then that $\rho_i=1/2$ for some $i$. We have that $i>0$ since $\rho_0=\rho\neq 1/2$. Thus, it must be that either $\rho_i=f_+(\rho_{i-1})$ or $\rho_i=f_-(\rho_{i-1})$; in the former case, we have that $\rho_{i-1}=1$, contradicting that  $1\notin S_\rho$, and in the latter case, we have that $\rho_{i-1}=-1$, contradicting that $\{\rho_i\}^n_{i=0}$ was a certificate of smallest length certifying that $-1\in S_\rho$.

To complete the proof, we only need to establish \eqref{eq:Srho}. First, we show that $0,1\notin S_\rho$. Observe that $\rho\neq 0,1$, so any certificate that $0,1\in S_\rho$ must have nonzero length. Further, the range of the functions $f_+,f_-$ excludes 0, which implies that $0\notin S_\rho$. Moreover, the only way that we can have $1\in S_\rho$ is if for some $x\in S_\rho$ it holds that $f_+(x)=1$ or $f_-(x)=1$. Both of these mandate that $x=0$, but $0\notin S_\rho$ as we just showed.

The remaining bit of \eqref{eq:Srho}, i.e., that $-1\in S_\rho$, will require more  effort to prove. As a starting point, note that from $\rho\in \mathbb{Q}$, we have  that $S_\rho\subseteq \mathbb{Q}$. Also, $S_\rho$ is nonempty since $\rho\in S_\rho$. Thus, there exists $z^*\in S_\rho$ such that $z^*=p/q$ where $p,q$ are integers such that $|p|+|q|$ is minimum. Since $|p|+|q|$ is minimum, it must be the case that $\gcd(p,q)=1$.

We first prove that $z^*\in T$; note, we already know that $z^*\neq 0,1$ since $0,1\notin S_\rho$ and $z^*\in S_\rho$, but keeping the values $0,1$ into consideration will be convenient for the upcoming argument. Namely, for the sake of contradiction, assume that $z^*\notin T$, which implies in particular that $z^*\neq 0,-1$. Since $z^*\in S_\rho$,  by \eqref{eq:oppositesign}, we obtain that $-z^*\in S_\rho$ as well. By switching to $-z^*$ if necessary, we may thus assume that $z^*$ is positive and hence that $p,q>0$, i.e., that both $p,q$ are positive integers. Since $z^*\neq 1$ (from $z^*\notin T$), we have that $p\neq q$. For each of the cases $p>q$ and $p<q$,  we obtain a contradiction to the minimality of $p+q$ by constructing $z'=p'/q'\in S_\rho$ with $p',q'$ positive integers such that $0<p'+q'<p+q$.

\begin{enumerate}[label=\textbf{Case \arabic*}.,leftmargin=*]
\item $p>q$. Since $z^*\neq 1,2$ (from $z^*\notin T$), we have that $p/q\neq 1$ and $f_-(p/q)=\frac{q}{q-p}\neq 0,-1$, so by \eqref{eq:oppositesign}  we have that
\[f_-(f_-(f_+(f_-(p/q))))=\frac{q}{p-q}.\]
Thus, letting $p'=q$ and $q'=p-q$ yields $z'=p'/q'\in S_\rho$ with $p'>0,q'>0$ and $0<p'+q'<p+q$.

\item $p<q$. Since $z^*\neq 0,1/2,1$ (from $z^*\notin T$), by \eqref{eq:fpinverse} we have that
\[f_-(f_-(f_+(f_-(f_-(p/q))))))=\frac{q-p}{p}.\]
Thus, letting $p'=q-p$ and $q'=p$ yields $z'=p'/q'\in S_\rho$ with $p'>0,q'>0$ and $0<p'+q'<p+q$.
\end{enumerate}
This concludes the proof that $z^*\in T$. In fact, we can now deduce easily that $-1\in S_\rho$. As noted earlier, we have that $z^*\neq 0,1$ as a consequence of $0,1\notin S_\rho$, so in fact $z^*\in\{-1,1/2,2\}$. If $z^*=-1$, then we automatically have that $-1\in S_\rho$ since $z^*$ was chosen to be in $S_\rho$. If $z^*=2$, then  we have that $2\in S_\rho$ and hence $f_-(2)=-1\in S_\rho$ as well. Finally, if $z^*=1/2$, we  have that $1/2\in S_\rho$ and hence $f_-(f_-(1/2))=-1\in S_\rho$.
Thus, it holds that $-1\in S_\rho$, which completes the proof of \eqref{eq:Srho}.

This concludes the proof of Lemma~\ref{lem:SequalsQ}.
\end{proof}

We are now ready to prove Lemma~\ref{lem:zero}.
\begin{lemzero}
\statelemzero
\end{lemzero}
\begin{proof}[Proof of Lemma~\ref{lem:zero}]
By Lemma~\ref{lem:minusone}, $(\Delta,\lambda)$ implements the activities $-1$ and $+1$. Thus, by Lemma~\ref{lem:implementingS}, we have that $(\Delta,\lambda)$ also implements the set of activities $\{\lambda z\mid z\in S\}$, where $S\subseteq \mathbb{R}$ is given in Definition~\ref{def:S}. By Lemma~\ref{lem:SequalsQ}, we have that $S=\mathbb{Q}$ and hence $(\Delta,\lambda)$ implements a dense set of activities in $\Reals$, as wanted (since $\lambda\neq 0$).
\end{proof}

\subsection{Proof of Lemma~\ref{lem:smalllambda}}\label{sec:darytree}
In this section, we give the proof of Lemma~\ref{lem:smalllambda}, which is the final missing piece that was used in the proof of Theorem~\ref{thm:precision}.

Recall from Section~\ref{sec:Rl} that the ratio $R_\lambda(G,v)$ is defined as $\Zout_{G,v}(\lambda)/Z_G(\lambda)$, whenever $Z_G(\lambda)\neq0$.  The following standard lemma gives a recursive procedure to compute $R_\lambda(G,v)$ and will thus be useful in studying the activities that $(\Delta,\lambda)$ implements.
\begin{lemma}\label{lem:recursion}
Let $\lambda\in\Rnz$. Let $G$ be a connected graph and let  $v$ be a vertex all of whose neighbours are in different components of $G\backslash v$.  Denote by $G_1,\hdots, G_d$ the connected components of $G\backslash v$ and by $v_1,\hdots,v_d$ the neighbours of $v$ in $G_1,\hdots, G_d$. Assume that $Z_{G_1}(\lambda),\hdots,Z_{G_d}(\lambda)\neq 0$.

Then, $Z_G(\lambda)=0$ iff $\prod^{d}_{i=1}R_\lambda(G_i,v_i)=-1/\lambda$. Further, if $Z_G(\lambda)\neq 0$, it holds that
\[R_{\lambda}(G,v)=f\big(R_\lambda(G_1,v_1),\hdots,R_\lambda(G_d,v_d)\big), \mbox{ where } f(x_1,\hdots,x_d):=\frac{1}{1+\lambda\prod^d_{i=1}x_i}.\]
\end{lemma}
\begin{proof}
For convenience, we drop the $\lambda$'s from notation. Using that $Z_{G_1},\hdots,Z_{G_d}\neq 0$, we have that
\[Z_G=\Zin_{G,v}+\Zout_{G,v}=\lambda\prod^d_{i=1}\Zout_{G_i,v_i}+\prod^{d}_{i=1}Z_{G_i}=\Big(\prod^{d}_{i=1}Z_{G_i}\Big)\Big(\lambda\prod^d_{i=1}R(G_i,v_i)+1\Big),\]
and thus $Z_G=0$ iff $\prod^{d}_{i=1}R(G_i,v_i)=-1/\lambda$. Also, we have
\[R(G,v)=\frac{\Zout_{G,v}}{\Zin_{G,v}+\Zout_{G,v}}=\frac{\prod^{d}_{i=1}(\Zin_{G_i,v_i}+\Zout_{G_i,v_i})}{\lambda\prod^d_{i=1}\Zout_{G_i,v_i}+\prod^{d}_{i=1}(\Zin_{G_i,v_i}+\Zout_{G_i,v_i})}=\frac{1}{1+\lambda\prod^d_{i=1}R(G_i,v_i)}.\qedhere\]
\end{proof}

We will also need the following technical lemma.
\begin{lemma}\label{lem:approx2q2}
Let $\lambda\in \mathbb{R}$. Then, for all $x\neq -1$, there exist  positive constants $\eta,M>0$ such that for all $x'$ with $|x-x'|\leq \eta$, it holds that
\[\Big|\frac{\lambda}{1+x}-\frac{\lambda}{1+x'}\Big|\leq M|x-x'|.\]
\end{lemma}
\begin{proof}
We may assume that $\lambda\neq 0$, otherwise the result is trivial. The proof is analogous to that of Lemma~\ref{lem:crudeapprox}.  In particular, since $x\neq -1$, there exists $\eta>0$ such that $1+x'\neq 0$ for all $x'$ such that  $|x-x'|\leq \eta$. Consider the function $f(y)=\lambda/(1+y)$ for $y$ in the interval $I=[x'-\eta,x'+\eta]$. Then, $f'(y)$ is well-defined and continuous in the interval $I$, so by letting $M=\max_{y\in I} |f'(y)|$, we obtain the result.
\end{proof}

We are now ready to prove Lemma~\ref{lem:smalllambda}, which we restate here for convenience.
\begin{lemsmalllambda}
\statelemsmalllambda
\end{lemsmalllambda}
\begin{proof}
For convenience, let $d:=\Delta-1$. Also, let $A:=-\lambda$ so that $A>0$; in fact, the condition $-1/4\leq \lambda<-\lambda^*(\Delta)$ translates into the bounds $1/4\geq A>d^d/(d+1)^{d+1}$.

For an integer  $h\geq 0$, let $T_h$ denote the $d$-ary tree of height $h$ and denote the root of the tree by $\rho$. For all $h$ such that $Z_{T_h}(\lambda)\neq 0$, let $x_h=R_\lambda(T_h,\rho)$.

\textbf{Part I.}  We show that there exists an $h$ such that $Z_{T_h}(\lambda)\neq 0$ and $x_h\geq (1/A)^{1/d}$.

For the sake of contradiction, assume otherwise. Then,
\begin{equation}\label{eq:contra23}
\mbox{ for all $h\geq 0$, either $Z_{T_h}(\lambda)=0$ or $x_h< (1/A)^{1/d}$}.
\end{equation}
In the following, we first exclude the possibility that $Z_{T_h}(\lambda)=0$ for some $h$, so that we can use the recursion from Lemma~\ref{lem:recursion} to study the range of the sequence $\{x_{h}\}^{\infty}_{h=0}$. In particular, assuming \eqref{eq:contra23}, we first  prove by induction that, for all $h\geq 0$, the following hold. (\eqref{eq:ZTh} is just used for the proof of \eqref{eq:recur}, later we will only appeal to \eqref{eq:recur}.)
\begin{gather}
Z_{T_h}(\lambda)\neq 0,\label{eq:ZTh}\\
x_h\in [0, (1/A)^{1/d}),\quad x_{h+1}=f(x_{h})\mbox { where } f(x)=\frac{1}{1-Ax^d}.\label{eq:recur}
\end{gather}
 For $h=0$, we have that $T_0$ is the single vertex graph, so $Z_{T_0}(\lambda)=1+\lambda \neq 0$ and hence $x_0=1/(1+\lambda)=1/(1-A)\geq 0$. Since $Z_{T_0}(\lambda)\neq 0$, \eqref{eq:contra23} yields that $x_0<(1/A)^{1/d}$.  For the induction step, assume that \eqref{eq:ZTh} and \eqref{eq:recur} hold for some integer $h$, we will prove them for $h+1$ as well. Since $T_{h+1}\backslash \rho$ consists of $d$ disconnected copies of $T_h$ and $Z_{T_h}(\lambda)\neq 0$ by \eqref{eq:ZTh}, we may apply Lemma~\ref{lem:recursion}. In particular, since $(1/A)^{1/d}> x_h\geq 0$ by \eqref{eq:recur}, we have that $-1/\lambda=1/A>(x_h)^d$ and hence the first part of Lemma~\ref{lem:recursion} yields that $Z_{T_{h+1}}(\lambda)\neq0$, so that $x_{h+1}$ is well-defined. Further, by \eqref{eq:contra23}, we obtain that $x_{h+1}<(1/A)^{1/d}$. Now, we note that  the second part of Lemma~\ref{lem:recursion} applies, so that $x_{h+1}=f(x_h)$ where the function $f$ is as in \eqref{eq:recur}. Since $x_h<(1/A)^{1/d}$, this in turn implies that $x_{h+1}\geq 0$ as well. This completes the induction.

Thus, assuming \eqref{eq:contra23}, we have established that the values $x_h$ are well-defined for all $h$ and that they satisfy the recursion in \eqref{eq:recur}. We will reach a contradiction to \eqref{eq:contra23} by showing that the sequence $x_h$ shoots over $(1/A)^{1/d}$. To do this, we will use that the sequence $x_h$ is increasing by proving that, for all $A>d^d/(d+1)^{d+1}$, it holds that
\begin{equation}\label{eq:sign}
f(x)>x \mbox{ for all }  x\in[0, (1/A)^{1/d}).
\end{equation}
To see this, note that
\[f(x)-x=\frac{1-x+Ax^{d+1}}{1-Ax^{d}},\]
and hence to show \eqref{eq:sign} it suffices to show that $g(x):=1-x+Ax^{d+1}>0$ for all $x\in[0, (1/A)^{1/d})$. Note that $g'(x)=(d+1)Ax^d-1$, so $g(x)\geq g(z_0)$ where $z_0$ satisfies $Az_0^{d}=1/(d+1)$. Then
\[g(z_0)=1+z_0(Az_0^d-1)=1-\frac{d z_0}{d+1}=1-\frac{d}{A^{1/d}(d+1)^{(d+1)/d}}>0,\]
where in the last inequality we used that $A>d^d/(d+1)^{d+1}$. By \eqref{eq:recur} and \eqref{eq:sign}, we obtain that the sequence $x_h$ is strictly increasing. Since $x_h\in[0, (1/A)^{1/d})$ it must converge to a limit $x^*\in [0, (1/A)^{1/d}]$  satisfying $f(x^*)=x^*$. By \eqref{eq:sign}, it must be the case that $x^*=(1/A)^{1/d}$ which is a contradiction to \eqref{eq:recur} since $f(x)\uparrow \infty$ as $x\uparrow (1/A)^{1/d}$. Thus, our assumption \eqref{eq:contra23} is false and, in particular, there is an integer $h\geq 0$ such that $Z_{T_h}(\lambda)\neq 0$ and $x_h\geq(1/A)^{1/d}$.

\textbf{Part II.} We next show how to use Part I to conclude the proof of the lemma. Let $\widehat{\lambda}:=\lambda (x_h)^{d-1}$. The key observation that will allow us to use the analysis of the paths is that $\widehat{\lambda}<-1/4$. Indeed, we have that
\begin{equation}\label{eq:lambdastarin}
\widehat{\lambda}\leq -A(1/A)^{(d-1)/d}= -A^{1/d}<-\frac{d}{(d+1)^{(d+1)/d}}<-1/4,
\end{equation}
where the last inequality holds for all $d\geq 2$. To utilise this, we need to construct graphs that simulate the hard core model with activity  $\widehat{\lambda}$ on paths.

In particular, let $P_n=(V,E)$ denote the path with $n$ vertices and $v$ be one of the endpoints of $P_n$. Let $G_n=(V_n,E_n)$ be the tree obtained from $P_n$ as follows. For each vertex $w\in V$ of the path, take $d-1$ distinct copies of the tree $T_h$ and connect $w$ to the roots of these trees. Note that the degree of the vertex $v$ in $G_n$ is $d=\Delta-1$, while every other vertex of $P_n$ which is not an endpoint has degree $d+1=\Delta$ in $G_n$.

 We claim that
\begin{equation}\label{eq:wPn}
\Zout_{P_n,v}\big(\widehat{\lambda}\big)=\frac{\Zout_{G_n,v}(\lambda)}{\big(Z_{T_h}(\lambda)\big)^{(d-1)n}}, \quad \Zin_{P_n,v}\big(\widehat{\lambda}\big)=\frac{\Zin_{G_n,v}(\lambda)}{\big(Z_{T_h}(\lambda)\big)^{(d-1)n}}, \quad Z_{P_n}\big(\widehat{\lambda}\big)=\frac{Z_{G_n}(\lambda)}{\big(Z_{T_h}(\lambda)\big)^{(d-1)n}}.
\end{equation}
Indeed, let $I$ be an independent set of $P_n$ and consider the set $\Omega_I$ of independent sets of $G_n$ whose restriction on $V$ coincides with $I$. Then, we have that
\begin{equation}\label{eq:tosumPn}
\sum_{I'\in \Omega_I}\lambda^{|I'|}=\prod_{u\in I}\lambda\big(\Zout_{T_h,\rho}(\lambda)\big)^{d-1}\prod_{u\notin I}\big(Z_{T_h}(\lambda)\big)^{d-1}=\big(Z_{T_h}(\lambda)\big)^{(d-1)n}\big(\widehat{\lambda}\big)^{|I|}.
\end{equation}
Observe also that the sets $\{\Omega_{I}\}_{I\in \mathcal{I}_{P_n}}$ form a partition of the set $\mathcal{I}_{G_n}$ of independent sets of $G_n$. Thus, summing \eqref{eq:tosumPn} over all $I\in \mathcal{I}_{P_n}$ such that $v\notin I$ yields the first equality in \eqref{eq:wPn}, summing \eqref{eq:tosumPn} over all $I\in \mathcal{I}_{P_n}$ such that $v\in I$ yields the second equality in \eqref{eq:wPn}, and, finally, summing \eqref{eq:tosumPn} over all  $I\in \mathcal{I}_{P_n}$ yields the third equality in \eqref{eq:wPn}.

Using the fact that $\widehat{\lambda}<-1/4$ by \eqref{eq:lambdastarin} and equation \eqref{eq:wPn}, we can now complete the proof of the lemma, by considering cases whether $\widehat{\lambda}\in\Bc$.

\vskip 0.2cm \noindent \textbf{Case 1.} $\widehat{\lambda}\in \Bc$. By Corollary~\ref{lem:bad}, there exists an integer $n\geq 1$ such that $Z_{P_n}\big(\widehat{\lambda}\big)=0$. Using the third equality in \eqref{eq:wPn}, we have that $G_n$ is a tree of maximum degree $\Delta$ such that $Z_{G_n}(\lambda)=0$. It follows by Lemma~\ref{lem:zero} that $(\Delta,\lambda)$ implements a dense set of activities in $\mathbb{R}$ as wanted.

\vskip 0.2cm \noindent \textbf{Case 2.} $\widehat{\lambda}\notin \Bc$. In this case, we have that for all $n\geq 1$ it holds that $Z_{P_n}\big(\widehat{\lambda}\big)\neq 0$ (this follows from Lemma~\ref{lem:closedformula} and  the definition \eqref{eq:bad} of the set $\Bc$; note that for $\lambda<-1/4$, there is a unique value of $\theta\in (0,\pi/2)$ such that $\lambda=-1/(2\cos\theta)^2$). Hence, we also have that  $\Zout_{P_n,v}\big(\widehat{\lambda}\big)=Z_{P_{n-1}}\big(\widehat{\lambda}\big)\neq 0$.\footnote{For $n=1$, we have $\Zout_{P_n,v}\big(\widehat{\lambda}\big)=1$ (the only independent set $I$ of $P_1$ such that $v\notin I$ is the empty one).} It follows from \eqref{eq:wPn} that $Z_{G_n}(\lambda),\Zout_{G_n,v}(\lambda)\neq 0$ as well.

In this case, our goal is to apply the path implementation of Lemma~\ref{lem:path} in combination with \eqref{eq:wPn}. Note however that the degree of $v$ in $G_n$ is $\Delta-1$ instead of one that is required for implementations, so we will add to the graph $G_n$ a new vertex $v'$ whose single neighbour is the vertex $v$. We denote this graph by $G_n'$, and note that, just as $G_n$, $G_n'$ is a tree of maximum degree $\Delta$. For all integers $n\geq 1$, using that $Z_{G_n}(\lambda),\Zout_{G_n,v}(\lambda)\neq 0$, we have that
\[\frac{\Zin_{G_n',v'}(\lambda)}{\Zout_{G_n',v'}(\lambda)}=\frac{\lambda \Zout_{G_n,v}(\lambda)}{\Zout_{G_n,v}(\lambda)+\Zin_{G_n,v}(\lambda)}=\frac{\lambda}{1+\frac{\Zin_{G_n,v}(\lambda)}{\Zout_{G_n,v}(\lambda)}}.\]
Let $\lambda_g$ be an activity that we wish to implement with error $\epsilon>0$. From the definition of a dense set, we may assume that $\lambda_g\neq \lambda$. Let $\lambda'=(\lambda-\lambda_g)/\lambda_g$ and note that $\lambda'\notin\{ -1,0\}$.  Let $\eta,M>0$ be the constants in Lemma~\ref{lem:approx2q2} when applied to $x=\lambda'$, so that for all $x'\in [\lambda'-\eta,\lambda'+\eta]$ it holds that
\begin{equation}\label{eq:sillyapprox}
\Big|\frac{\lambda}{1+x'}-\frac{\lambda}{1+\lambda'}\Big|=\Big|\frac{\lambda}{1+x'}-\lambda_g\Big|\leq M|x'-\lambda'|.
\end{equation}
Further, let $\epsilon':=\min\{\eta,\epsilon/M\}$. By the path implementation of Lemma~\ref{lem:path}, there exists $n$ such that
\begin{equation*}
\Big|\frac{\Zin_{P_n,v}\big(\widehat{\lambda}\big)}{\Zout_{P_n,v}\big(\widehat{\lambda}\big)}-\lambda'\Big|\leq \epsilon', \mbox{ and hence by \eqref{eq:wPn} we obtain   \ \ } \Big|\frac{\Zin_{G_n,v}(\lambda)}{\Zout_{G_n,v}(\lambda)}-\lambda'\Big|\leq \epsilon',
\end{equation*}
Since $\epsilon'\leq \eta$ and $\epsilon'\leq \epsilon/M$, by \eqref{eq:sillyapprox} we obtain
\begin{equation*}
\Big|\frac{\Zin_{G_n',v'}(\lambda)}{\Zout_{G_n',v'}(\lambda)}-\lambda_g\Big|=\Big|\frac{\lambda}{1+\frac{\Zin_{G_n,v}(\lambda)}{\Zout_{G_n,v}(\lambda)}}-\frac{\lambda}{1+\lambda'}\Big|\leq M\Big|\frac{\Zin_{G_n,v}(\lambda)}{\Zout_{G_n,v}(\lambda)}-\lambda'\Big|\leq \epsilon.
\end{equation*}
Thus, $G_n'$ with terminal $v'$ implements the activity $\lambda_g$ with error $\epsilon>0$, as desired.

This completes the proof of Lemma~\ref{lem:smalllambda}.
\end{proof}

\bibliographystyle{plain}
\bibliography{\jobname}

\clearpage

\section{Appendix: Application --- an intractability result}
 
Section~\ref{sec:alg} already 
discussed one intractability result that uses our Theorem~\ref{thm:precision}.
In particular,
 \cite[Theorem 1]{OurComplex} shows
that, for any 
$\epsilon>0$,
$\Delta\geq 3$ and $\lambda<-\lambda^*(\Delta)$,
it is \#P-hard to approximate the absolute value of $Z_G(\lambda)$   within a factor of~$2^{n^{1-\epsilon}}$, where $n$ is the number of vertices of~$G$
(which is bipartite and has maximum degree at most~$\Delta$). 
In this appendix, we again use Theorem~\ref{thm:precision} to
obtain an incomparable hardness result, Theorem~\ref{thm:neglambda},
 which shows that,
 with the same restrictions on~$\Delta$, $\lambda$ and~$G$, it is NP-hard to approximate $|Z_G(\lambda)|$ within an exponential factor.\footnote{Not also the earlier NP-hardness result of~\cite[Theorem 4.4]{HSVV1} 
 which 
 shows, for $\Delta\geq 62$ and $\lambda<-39/\Delta$, that a PTAS would imply NP=RP. }

\subsection{Preliminaries on antiferromagnetic 2-spin systems on $\Delta$-regular graphs}
In our setting, where every vertex has degree at most $\Delta$, an   implementation consumes one of the $\Delta$ slots that a vertex has available to connect to other vertices. This is particularly problematic for the case where $\Delta=3$. In the following we circumvent this problem by constructing suitable binary gadgets, so that we can use inapproximability results for computing the partition function of antiferromagnetic 2-spin systems on $\Delta$-regular graphs.

An antiferromagnetic 2-spin system (without an external field) is specified by two parameters $\beta,\gamma>0$ such that $\beta\gamma<1$. Let $\Mb=\{M_{ij}\}_{i,j\in\{0,1\}}$ be the matrix $\big[\begin{smallmatrix} \beta&1\\1&\gamma\end{smallmatrix}\big]$. For a graph $H=(V,E)$, configurations of the 2-spin system are assignments $\sigma:V\rightarrow\{0,1\}$ and the weight of a configuration $\sigma$ is given by $w_{H,\beta,\gamma}(\sigma)=\prod_{\{u,v\}\in E} M_{\sigma(u),\sigma(v)}$. The partition function of $H$ is then given by
\[Z_{H,\beta,\gamma}=\sum_{\sigma:V\rightarrow\{0,1\}} w_{H,\beta,\gamma}(\sigma)= \sum_{\sigma:V\rightarrow\{0,1\}} \prod_{\{u,v\}\in E}M_{\sigma(u),\sigma(v)}.\]
For positive parameters $\beta,\gamma$ and $c>1$, we consider the following computational problem, where the input is a 3-regular graph $H$.
\prob{ $\TwoSpin{\beta}{\gamma}$.}
{ An $n$-vertex graph $H$ which is $3$-regular.}
{  A number $\hat{Z}$ such that $c^{-n}Z_{H,\beta,\gamma}\leq \hat{Z}\leq c^n Z_{H,\beta,\gamma}$.}

The case $\beta=\gamma<1$ corresponds to the well-known (antiferromagnetic) Ising model. As a corollary of results of Sly and Sun \cite{SlySun} (see also \cite{GSVIsing}), it is known that, for $0<\beta=\gamma<1/3$, there exists $c>1$ such  that $\TwoSpin{\beta}{\beta}$ is $\NP$-hard, i.e.,  approximating the partition function $Z_{G,\beta,\beta}$ of the Ising model on $3$-regular graphs $H$ is $\NP$-hard, even within an exponential factor.\footnote{The inapproximability result for the Ising model holds for general degrees $\Delta\geq 3$ in the regime $0<\beta=\gamma<(\Delta-2)/\Delta$. While we do not prove it here (since we only need the result for $\Delta=3$), Lemma~\ref{lem:spinhard} also holds for general degrees $\Delta\geq 3$ in the square $0<\beta,\gamma<(\Delta-2)/\Delta$.}   The following lemma is somewhat less known but follows easily from the results of \cite{SlySun}.
\begin{lemma}\label{lem:spinhard}
Let $\Delta=3$ and $\beta,\gamma$ be such that $0<\beta,\gamma<1/3$. Then, there exists $c>1$ such  that $\TwoSpin{\beta}{\gamma}$ is $\NP$-hard.
\end{lemma}
\begin{proof} 
Sly and Sun \cite{SlySun} give a sufficient condition on the range of $\beta,\gamma$ such that the conclusion of the lemma holds (the condition is in fact tight apart, perhaps, from certain boundary cases). Our goal is thus to verify that all $\beta,\gamma$ in the square $0<\beta,\gamma<1/3$ lie within the range where the result of \cite{SlySun} applies.

The condition in \cite{SlySun} asks that, for the 2-spin system specified by $\beta$ and $\gamma$, the 3-regular tree exhibits non-uniqueness. This is somewhat implicit for our purposes, so we will instead use the following algebraic criterion which is well-known in the area (see, e.g., \cite[Section 3.1]{GGhypergraphs} for a detailed discussion). The 3-regular tree has non-uniqueness (for the 2-spin system specified by $\beta,\gamma$) iff the equations
\begin{equation}\label{eq:nonuniqueness}
x=\Big(\frac{\beta y+1}{y+\gamma}\Big)^2, \quad y=\Big(\frac{\beta x+1}{x+\gamma}\Big)^2
\end{equation}
admit a solution $x,y>0$ with $x\neq y$ (note, there is always a solution with $x=y>0$). To verify this, consider the quadratic equation
\begin{equation}\label{eq:quadratic}
\begin{gathered}
A z^2+B z+C=0, \mbox{ where }\\
A:=(\beta^2+\gamma)^2,\quad B:=-1 + (\beta^2 + 2 \gamma) (2 \beta + \gamma^2),\quad C:=(\beta+\gamma^2)^2.
\end{gathered}
\end{equation}
We will show that for all $0<\beta,\gamma<1/3$, the quadratic equation admits two solutions $z_1,z_2>0$ with $z_1\neq z_2$. Then, we will show that $x=z_1$ and $y=z_2$ satisfies \eqref{eq:nonuniqueness}, thus verifying the condition of non-uniqueness on the 3-regular tree and, consequently, proving the lemma by applying the result of \cite{SlySun}.

So, suppose that $\beta,\gamma$ are such that $0<\beta,\gamma<1/3$. Note that the discriminant of the equation \eqref{eq:quadratic} is strictly positive,  since
\[B^2-4AC=(1-\beta \gamma)^2(1-4 \beta^3-6 \beta \gamma-3 \beta^2 \gamma^2-4 \gamma^3)>0,\]
where the last inequality follows from $0<\beta,\gamma<1/3$. Since $A>0, B<0, C>0$, we conclude that \eqref{eq:quadratic} has two positive solutions $z_1,z_2>0$ which satisfy $z_1\neq z_2$.

It remains to show that $z_1,z_2$ satisfy \eqref{eq:nonuniqueness}. We only need to show the first equality since the second follows by swapping the roles of $z_1,z_2$. Note that $z_1+z_2=-B/A$, so we only need to show that
\[-\frac{B}{A}-z_1=\Big(\frac{\beta z_1 +1}{z_1+\gamma}\Big)^2 \mbox{ or, by multiplying out, } Az_1(z_1+\gamma)^2+ A(\beta z_1+1)^2+B(z_1+\gamma)^2=0.\]
Using the values of $A,B,C$ in \eqref{eq:quadratic} and in particular that
$
A+B \gamma^2=C(\beta^2+2\gamma)$ and $A(2\beta+\gamma^2)=C+B\beta^2$,
we obtain the factorisation
$(z_1+\beta^2 + 2 \gamma)(A z^2_1+B z_1+C)=0$,
which is clearly true since $z_1$ is a root of \eqref{eq:quadratic}.

This concludes the proof of Lemma~\ref{lem:spinhard}.
\end{proof}

The following lemma will be used in  the proof of Theorem~\ref{thm:neglambda} to specify the activities that we need to implement to utilise the inapproximability result of Lemma~\ref{lem:spinhard}. It allows us to use the graph in Figure~\ref{fig:antiferro} as a binary gadget to simulate a 2-spin system with parameters $\beta,\gamma$.

\begin{lemma}\label{lem:betagamma}
Let $\lambda<0$. Then, there exist $\lambda_1',\lambda_2'$  such that
\begin{equation}\label{eq:lambda12}
-2-\frac{1}{3}|\lambda|^{1/3}<\lambda_1'<\min\Big\{-2,-2-\frac{|\lambda|^{2/3}-1}{3|\lambda|^{1/3}+1}\Big\},\qquad -1 <\lambda_2'<-1-\frac{\lambda_1'(\lambda_1'+2+\frac{1}{3}|\lambda|^{1/3})}{1+\frac{1}{3}|\lambda|^{1/3}}.
\end{equation}
For all $\lambda_1',\lambda_2'$ satisfying \eqref{eq:lambda12}, the following parameters $\beta,\gamma$ (defined in terms of $\lambda,\lambda_1',\lambda_2'$)
\begin{equation}\label{eq:betagamma}
\beta=-\frac{(\lambda'_1+1)^2+\lambda'_2}{|\lambda|^{1/3}\big(1+\lambda_1'+\lambda_2'\big)},\quad \gamma=-\frac{|\lambda|^{1/3}(1+\lambda_2')}{1+\lambda_1'+\lambda_2'}.
\end{equation}
satisfy  $0<\beta, \gamma<1/3$.
\end{lemma}
\begin{proof} 
Note that for all $\lambda<0$, we have 
\begin{equation}\label{eq:c2ecc2ef2}
-2-\frac{1}{3}|\lambda|^{1/3}<-2, \qquad -2-\frac{1}{3}|\lambda|^{1/3}<-2-\frac{|\lambda|^{2/3}-1}{3|\lambda|^{1/3}+1},
\end{equation}
so the interval
\begin{equation}\label{eq:inI1}
I_1:=\bigg(-2-\frac{1}{3}|\lambda|^{1/3},\min\Big\{-2,-2-\frac{|\lambda|^{2/3}-1}{3|\lambda|^{1/3}+1}\Big\}\bigg).
\end{equation}
has nonzero length. Thus, choosing any  $\lambda_1'\in I_1$ satisfies the first inequality in \eqref{eq:lambda12}. Also, for $\lambda_1'\in I_1$ we have that $\lambda_1'<0$ and $\lambda_1'+2+\frac{1}{3}|\lambda|^{1/3}>0$, so the interval
\begin{equation}\label{eq:inI2}
I_2:=\Big(-1,-1-\frac{\lambda_1'(\lambda_1'+2+\frac{1}{3}|\lambda|^{1/3})}{1+\frac{1}{3}|\lambda|^{1/3}}\Big)
\end{equation}
has nonzero length as well. Thus, by first choosing $\lambda_1'\in I_1$ and then $\lambda_2'\in I_2$, we see that $\lambda_1',\lambda_2'$ satisfy \eqref{eq:lambda12}.

Next, we show that the parameters $\beta,\gamma$ in \eqref{eq:betagamma} satisfy the desired inequalities whenever $\lambda_1'\in I_1$ and $\lambda_2'\in I_2$. We first prove that $\beta,\gamma>0$ by showing the following inequalities:
\begin{gather}
(\lambda_1'+1)^2+\lambda_2'>0,\label{eq:ineq1c}\\
1+\lambda_2'>0,\label{eq:ineq1a}\\
1+\lambda_1'+\lambda_2'<0.\label{eq:ineq1b}
\end{gather}
The inequality in \eqref{eq:ineq1a} is an immediate consequence of $\lambda_2'\in I_2$. Inequality \eqref{eq:ineq1c} follows from the expansion
\begin{equation}\label{eq:expansion}
(\lambda_1'+1)^2+\lambda_2'=\lambda_1'(\lambda_1'+2)+\lambda_2'+1,
\end{equation}
and noting that $\lambda_1'(\lambda_1'+2)>0$ (from $\lambda_1'\in I_1$) and $\lambda_2'+1>0$ (from $\lambda_2'\in I_2$). Finally, for \eqref{eq:ineq1b}, we have that $\lambda_2'+1<-\frac{\lambda_1'(\lambda_1'+2+\frac{1}{3}|\lambda|^{1/3})}{1+\frac{1}{3}|\lambda|^{1/3}}$ (from $\lambda_2'\in I_2$) and hence
\[1+\lambda_1'+\lambda_2'<\lambda_1'\Big(1-\frac{\lambda_1'+2+\frac{1}{3}|\lambda|^{1/3}}{1+\frac{1}{3}|\lambda|^{1/3}}\Big)=-\frac{\lambda_1'(\lambda_1'+1)}{1+\frac{1}{3}|\lambda|^{1/3}}<0,\]
where the last inequality follows from $\lambda_1'<-2$. Thus, we have shown that $\beta,\gamma>0$.

Next, we show that $\beta,\gamma<1/3$.  To show that $\beta<1/3$, using \eqref{eq:ineq1b}, we only need to show that
\[(\lambda_1'+1)^2+\lambda_2'<-\frac{1}{3}|\lambda|^{1/3}(1+\lambda_1'+\lambda_2'), \mbox{ or equivalently } (\lambda_2'+1)(1+\frac{1}{3}|\lambda|^{1/3})<-\lambda_1'(\lambda_1'+2+\frac{1}{3}|\lambda|^{1/3}),\]
which is true since $\lambda_2'\in I_2$. To show that $\gamma<1/3$, using \eqref{eq:ineq1b} again, we see that the inequality $\gamma<1/3$ is equivalent to
\[-3|\lambda|^{1/3}(1+\lambda_2')>1+\lambda_1'+\lambda_2', \mbox{ or equivalently } \lambda_2'<-1-\frac{\lambda_1'}{1+3|\lambda|^{1/3}},\]
Since $\lambda_1'<0$ and $\lambda_2'<-1-\frac{\lambda_1'(\lambda_1'+2+\frac{1}{3}|\lambda|^{1/3})}{1+\frac{1}{3}|\lambda|^{1/3}}$, we only need to show that
\[\frac{\lambda_1'+2+\frac{1}{3}|\lambda|^{1/3}}{1+\frac{1}{3}|\lambda|^{1/3}}<\frac{1}{1+3|\lambda|^{1/3}},\mbox{ or } (\lambda_1'+2)(1+3|\lambda|^{1/3})+\frac{1}{3}|\lambda|^{1/3}(1+3|\lambda|^{1/3})<1+\frac{1}{3}|\lambda|^{1/3},\]
which is true since,  from $\lambda_1'\in I_1$, we have $(\lambda_1'+2)(1+3|\lambda|^{1/3})<-(|\lambda|^{2/3}-1)$.

Thus, we have shown that $0<\beta,\gamma<1/3$, thus completing the proof of Lemma~\ref{lem:betagamma}.
\end{proof}

\subsection{The reduction}
The reduction to obtain Theorem~\ref{thm:neglambda} uses a binary gadget to simulate an antiferromagnetic 2-spin system on 3-regular graphs, i.e., we will replace every edge of a 3-regular graph $H$ with a suitable graph $B$ which has two special vertices to encode the edge. The gadget $B$ is given in Figure~\ref{fig:antiferro}, the two special vertices are $v_1,v_2$. Note that the gadget $B$ has nonuniform activities but this will be compensated for later by invoking Lemma~\ref{lem:nonuniform}.
\begin{lemma}\label{lem:reduction}
Let $\lambda<0$ and $\lambda_1',\lambda_2'\in \mathbb{R}$ satisfy \eqref{eq:lambda12}. Then, for $\beta,\gamma$ as in \eqref{eq:betagamma}, the following holds.
For every 3-regular graph $H=(V_H,E_H)$ we can construct in linear time a bipartite graph $G=(V_G,E_G)$ of maximum degree $3$ and specify an activity vector $\lambdab=\{\lambda_v\}_{v\in V}$ on $G$ such that
\begin{enumerate}
\item \label{it:simulate} $Z_{H,\beta,\gamma}= Z_{G}(\lambdab)/C^{|E_H|}$, where $C:=-|\lambda|^{1/3}\big(\lambda_1'+\lambda_2'+1\big)>0$.\footnote{The fact that $C$ is positive follows from \eqref{eq:ineq1b}.}
\item \label{it:degree3} For every vertex $v$ of $G$, it holds that $\lambda_v\in\{\lambda,\lambda_1',\lambda_2'\}$. Moreover, if $\lambda_v\neq \lambda$, then $v$ has degree two in $G$.
\end{enumerate}
\end{lemma}
\begin{proof}
Let $H=(V_H,E_H)$ be a 3-regular graph.
\begin{figure}[h]
\begin{center}
{
\tikzset{lab/.style={circle,draw,inner sep=0pt,fill=none,minimum size=5mm}}
\begin{tikzpicture}[xscale=1,yscale=1]
\draw (0,0) node[lab, label={below:$\lambda_{v_1}=-|\lambda|^{1/3}$}] (1) {$v_1$};
\draw (0,-3) node[lab, label={below:$\lambda_{v_2}=-|\lambda|^{1/3}$}] (2) {$v_2$};
\draw (6,-1.5) node[lab, label={right:$\lambda_{z}=\lambda_2'$}] (5) {$z$};
\draw (3,-0.75) node[lab, label={above:$\lambda_x=\lambda_1'$}] (3) {$x$};
\draw (3,-2.25) node[lab, label={below:$\lambda_{y}=\lambda_1'$}] (4) {$y$};
\draw (1) -- (3) -- (5);
\draw (2) -- (4) -- (5);
\end{tikzpicture}
}
\end{center}
\caption{The binary gadget $B=(U,F)$ used in Lemma~\ref{lem:reduction} to simulate an antiferromagnetic 2-spin system on 3-regular graphs. The gadget $B$ is used to encode  the edges of a 3-regular graph $H$. In particular, every edge $e=\{h_1,h_2\}$ of $H$ gets replaced by a distinct copy of $B$, with the vertices $v_1,v_2$ of $B$ getting identified with the vertices $h_1,h_2$ of $H$, respectively.}
\label{fig:antiferro}
\end{figure}

To construct the graph $G$, we will use the graph $B=(U,F)$ in Figure~\ref{fig:antiferro}; the vertices $v_1,v_2$ of $B$ will be used for connections. Roughly, the graph $G=(V_G,E_G)$ is constructed by replacing every edge $\{h_1,h_2\}$ of $H$ with a distinct copy of $B$ and identifying the vertex $v_1$ of $B$ with the vertex $h_1$ of $H$ and the vertex $v_2$ of $B$ with the vertex $h_2$ of $H$. The identification of the vertices $v_1,v_2$ with the vertices $h_1,h_2$ is done so that $V_H\subseteq V_G$, i.e., vertices in $H$ retain their labelling in $G$. Note that $B$ is symmetric with respect to $v_1,v_2$ and hence the ordering of the vertices $v_1,v_2$ and $h_1,h_2$ does not matter in the construction.

To give explicitly the construction of the graph $G$, for every edge $e=\{h_1,h_2\}\in E_H$, take a distinct copy of $B$. We will denote by $B^{(e)}=(U^{(e)},F^{(e)})$ the copy of $B$ corresponding to the edge  $e$ of $H$ and, for a vertex $u\in U$, we denote by $u^{(e)}$ the copy of the vertex of $u$ in the copy $B^{(e)}$. As noted earlier, we relabel $v^{(e)}_1$ to $h_1$ and $v^{(e)}_2$ to $h_2$. The graph $G=(V_G,E_G)$ is then given by
\begin{equation*}
V_G=\bigcup_{e\in E_H} U^{(e)}, \quad E_G=\bigcup_{e\in E_H} F^{(e)}.
\end{equation*}
We next specify an activity vector $\lambdab$ on $G$. Every vertex $v\in V_G\backslash V_H$ is the image of a vertex $u$ in $B$, and inherits the activity from its image $u$ in $B$ (cf. Figure~\ref{fig:antiferro} for the specification of the activities in $B$). Every vertex $h\in V_H$ is the image of three vertices whose activities in the graph $B$ were equal to $-|\lambda|^{1/3}$; since these three vertices were identified with $h$, we set the activity of the vertex $h$ to equal $(-|\lambda|^{1/3})^3=\lambda$ (our argument later will formally justify that multiplying the activities is indeed the right way to account for the effect of identification). Formally, the activity vector $\lambdab=\{\lambda_v\}_{v\in V_G}$ is given by
\begin{equation*}
\begin{gathered}
\forall h\in V_H:\quad\lambda_{h}=\lambda,\\
\forall e\in E_H:\quad  \lambda_{x^{(e)}}=\lambda_{y^{(e)}}=\lambda_1', \quad \lambda_{z^{(e)}}=\lambda_2',
\end{gathered}
\end{equation*}
where recall that the activities $\lambda_1',\lambda_2'$ satisfy \eqref{eq:lambda12}. It is now immediate that the graph $G$ has maximum degree three and that the activity vector $\lambdab$ satisfies Item~\ref{it:degree3} of the lemma statement. Moreover, $G$ is bipartite (every cycle in $G$ corresponds to a cycle in $H$; further, every cycle in $H$ maps to an even-length cycle in $G$ since the edge gadget $B$ is an even-length path).

To finish the proof of the lemma, it remains to establish Item~\ref{it:simulate}, i.e., to connect the partition functions $Z_{H,\beta,\gamma}$ and $Z_G(\lambdab)$, where the parameters $\beta,\gamma$ are given in \eqref{eq:betagamma}. 

Let $\sigma:V_H\rightarrow\{0,1\}$ be a $\{0,1\}$-assignment to the vertices of $H$. Let $\Omega_\sigma\subseteq \mathcal{I}_G$ be the set of independent sets of $G$ whose restriction on $H$ coincides with the set of vertices which are assigned the spin 1 under $\sigma$, i.e.,
\[\Omega_\sigma:=\{I\in \mathcal{I}_G\mid V_H\cap I=\sigma^{-1}(1)\}.\]
We will show that
\begin{equation}\label{eq:Hsigma}
w_{H,\beta,\gamma}(\sigma)=\frac{\sum_{I\in \Omega_\sigma}\prod_{v\in I} \lambda_v}{C^{|E_H|}}, \mbox{ where } C=-|\lambda|^{1/3}\big(\lambda_1'+\lambda_2'+1\big)>0.
\end{equation}
Note that the sets $\{\Omega_\sigma\}_{\sigma:V\rightarrow \{0,1\}}$ form a partition of the set $\mathcal{I}_G$, so adding \eqref{eq:Hsigma} over all $\sigma:V_H\rightarrow \{0,1\}$ gives that $Z_{H,\beta,\gamma}=Z_{G}(\lambdab)/C^{|E_H|}$, as wanted for Item~\ref{it:degree3}. Thus, we focus on proving \eqref{eq:Hsigma}.

To calculate the aggregate weight of independent sets in $\Omega_\sigma$, we first observe that the graph induced by $V_G\backslash V_H$ consists of $|E_H|$ disconnected copies of the graph $B\backslash \{v_1,v_2\}$. Thus, for each edge $e\in E_H$, we need to calculate the weight of independent sets that are consistent with the assignment $\sigma$ on the vertices $v^{(e)}_1,v^{(e)}_2$ of $B^{(e)}$. Then, to compute the aggregate weight of independent sets in $\Omega_\sigma$, we only need to multiply these quantities over all $e\in E_H$. Note that, for an independent set $I$ in $\Omega_\sigma$, a vertex $h$ in $V_H$ such that $\sigma(h)=1$ contributes a factor of $\lambda$ in the weight of $I$; a convenient way to account for this factor $\lambda$ is to split it into the three edges incident to $h$ by setting the activities of $v^{(e)}_1,v^{(e)}_2$ in $B^{(e)}$ equal to $-|\lambda|^{1/3}$. Then, when we multiply over $e\in E_H$, $h$ contributes in total a factor $(-|\lambda|^{1/3})^3=\lambda$, just as it should. In light of this, let
\begin{gather*}
Z_{00}=\sum_{I\in \mathcal{I}_B;\, v_1\notin I, v_2\notin I}\,\prod_{v\in I}\lambda_v,\qquad  Z_{11}=\sum_{I\in \mathcal{I}_B;\, v_1\in I, v_2\in I}\,\prod_{v\in I}\lambda_v\\
Z_{01}=\sum_{I\in \mathcal{I}_B;\, v_1\notin I, v_2\in I}\,\prod_{v\in I}\lambda_v,\qquad Z_{10}=\sum_{I\in \mathcal{I}_B;\, v_1\in I, v_2\notin I}\,\prod_{v\in I}\lambda_v
\end{gather*}
Note that $Z_{01}=Z_{10}$ since the graph $B$ is symmetric with respect to $v_1,v_2$.  Now, denote by $e_{00}, e_{11}, e_{01}$ the number of edges $\{h_1,h_2\}\in E_H$ such that $\sigma(h_1)=\sigma(h_2)=0$, $\sigma(h_1)=\sigma(h_2)=1$, and $\sigma(h_1)\neq\sigma(h_2)$, respectively. Then, we have that
\[\sum_{I\in \Omega_\sigma}\prod_{v\in I} \lambda_v=(Z_{00})^{e_{00}}(Z_{11})^{e_{11}}(Z_{01})^{e_{01}}=(Z_{01})^{|E_H|}\Big(\frac{Z_{00}}{Z_{01}}\Big)^{e_{00}}\Big(\frac{Z_{11}}{Z_{01}}\Big)^{e_{11}}.\]
Equation \eqref{eq:Hsigma} will thus follow by showing that for $\beta,\gamma$ as in \eqref{eq:betagamma} and $C$ as in \eqref{eq:Hsigma}, it holds that
\begin{equation}\label{eq:spoon}
Z_{01}=C, \quad \beta=\frac{Z_{00}}{Z_{01}}, \quad \gamma=\frac{Z_{11}}{Z_{01}}.
\end{equation}
We will justify \eqref{eq:spoon} by giving explicit expressions for $Z_{00},Z_{11},Z_{01}$ in terms of $\lambda,\lambda_1',\lambda_2'$; we give the derivation for $Z_{00}$, the other  quantities can be handled similarly (refer to Figure~\ref{fig:antiferro} for the following). For $Z_{00}$, we need only to consider independent sets $I\in \mathcal{I}_B$ such that $v_1,v_2\notin I$. Then, we consider cases whether $x,y\in I$. If $x,y\in I$, then $z\notin I$ and therefore the aggregate weight of  independent sets $I$ with $x,y\in I$ (and $v_1,v_2\notin I$) is given by $\lambda_x\lambda_y$. Similarly, the weight of independent sets $I$ such that $x\in I$ but $y\notin I$ is given by $\lambda_x$. The remaining cases ($x\notin I,y\in I$ and $x\notin I,y\notin I$) can be computed analogously. In this way, we obtain
\begin{align*}
Z_{11}&=\lambda_{v_1}\lambda_{v_2}(\lambda_z+1)=|\lambda|^{2/3}(1+\lambda_2'),\\[0.2cm]
Z_{01}&=\lambda_{v_2}\big(\lambda_x+\lambda_z+1\big)=-|\lambda|^{1/3}\big(\lambda_1'+\lambda_2'+1\big)=C,\\[0.2cm]
Z_{00}&=\lambda_x\lambda_y+\lambda_x+\lambda_y+\lambda_z+1=(\lambda'_1+1)^2+\lambda'_2.
\end{align*}
To conclude the validity of \eqref{eq:spoon}, it remains to juxtapose these expressions with the expressions of $\beta,\gamma$ given in \eqref{eq:betagamma} and verify that they are identical, which is indeed the case.

This completes the proof of \eqref{eq:Hsigma} and thus the proof of Lemma~\ref{lem:reduction}.
\end{proof}

\subsection{The inapproximability result}

 To formally state our result, we define the
following problem which has three parameters---the activity
$\lambda$, a degree bound $\Delta$, and a value $c>1$ which
specifies the desired accuracy of the approximation. \prob{
$\Hardcore$.} { An $n$-vertex bipartite graph $G$ with maximum degree at
most $\Delta$.} {  A number $\widehat{Z}$ such that
$c^{-n}|Z_{G}(\lambda)|\leq \big|\widehat{Z}\big|\leq
c^n|Z_{G}(\lambda)|$.}

\newcommand{\statethmneglambda}{Let $\Delta\geq 3$ and $\lambda<-\lambda^*(\Delta)$. Then there exists a constant $c>1$ such that   $\Hardcore$ is $\NP$-hard, i.e., it is $\NP$-hard to approximate $|Z_G(\lambda)|$ on bipartite graphs $G$ of maximum degree at most $\Delta$, even within an exponential factor.}
\begin{theorem}\label{thm:neglambda}
\statethmneglambda
\end{theorem}

\begin{proof}
We first specify two activities $\lambda_1',\lambda_2'$ that $(\Delta,\lambda)$ implements, which further satisfy the condition \eqref{eq:lambda12} of Lemma~\ref{lem:betagamma}. 
Let $I_1$ be the following interval (which was considered in the proof of Lemma~\ref{lem:betagamma})
\begin{equation*}\tag{\ref{eq:inI1}}
I_1:=\bigg(-2-\frac{1}{3}|\lambda|^{1/3},\min\Big\{-2,-2-\frac{|\lambda|^{2/3}-1}{3|\lambda|^{1/3}+1}\Big\}\bigg),
\end{equation*}
and recall that $I_1$ has nonzero length for all $\lambda<0$, cf. \eqref{eq:c2ecc2ef2}. Hence, there exist $l_1\in I_1$ and $\epsilon>0$ such that $[l_1-\epsilon,l_1+\epsilon]\subset I_1$. By 
our main result, Theorem~\ref{thm:precision}, there is a bipartite graph $G_1$ of maximum degree $\Delta$ with terminal $v_1$ that implements $l_1$ with accuracy $\epsilon$. Let $\lambda_1':=\Zin_{G_1,v_1}(\lambda)/\Zout_{G_1,v_1}(\lambda)$, so that $G_1$ with terminal $v_1$ implements $\lambda_1'$. Since  $[l_1-\epsilon,l_1+\epsilon]\subset I_1$, we have that $\lambda_1'\in I_1$. Note that $\lambda_1'$ can be computed by brute force in constant time (since $G_1$ is a fixed graph). Let $I_2$ be the following interval (which was also considered in the proof of Lemma~\ref{lem:betagamma})
\begin{equation*}\tag{\ref{eq:inI2}}
I_2:=\Big(-1,-1-\frac{\lambda_1'(\lambda_1'+2+\frac{1}{3}|\lambda|^{1/3})}{1+\frac{1}{3}|\lambda|^{1/3}}\Big).
\end{equation*}
By an analogous argument (the fact that $I_2$ has nonzero length for $\lambda_1'\in I_1$ is proved in Lemma~\ref{lem:betagamma}), we can specify a bipartite graph $G_2$ of maximum degree $\Delta$ with terminal $v_2$ that implements an activity $\lambda_2'\in I_2$. By construction, the bipartite graphs $G_1,G_2$ implement the activities $\lambda_1',\lambda_2'$, respectively, and $\lambda_1',\lambda_2'$ satisfy the condition \eqref{eq:lambda12} of Lemma~\ref{lem:betagamma}, as wanted.  For later use, set
\begin{equation}\label{eq:C1C2}
C_1:=\Zout_{G_1,v_1}(\lambda), \quad C_2:=\Zout_{G_2,v_2}(\lambda),
\end{equation}
and note that $C_1,C_2$ are also explicitly computable constants.

Let $\beta,\gamma$ be the parameters given by \eqref{eq:betagamma}. By Lemma~\ref{lem:betagamma}, it holds that $0<\beta,\gamma<1/3$. Thus, by Lemma~\ref{lem:spinhard}, there exists $c>1$ such  that $\TwoSpin{\beta}{\gamma}$ is $\NP$-hard. We will use Lemmas~\ref{lem:nonuniform} and~\ref{lem:reduction} to reduce $\TwoSpin{\beta}{\gamma}$ to $\HardCore{c'}$ for some constant $c'>1$.

Let $H$ be a 3-regular graph which is an input graph to the problem $\TwoSpin{\beta}{\gamma}$. By Lemma~\ref{lem:reduction}, we can construct in linear time a bipartite graph $G$ of maximum degree $3$ and specify an activity vector $\lambdab=\{\lambda_v\}_{v\in V}$ on $G$ such that
\begin{enumerate}
\item \label{it:ws34522} $Z_{H,\beta,\gamma}= Z_{G}(\lambdab)/C^{|E_H|}$, where $C:=-|\lambda|^{1/3}\big(\lambda_1'+\lambda_2'+1\big)>0$.
\item \label{it:ws234w} For every vertex $v$ of $G$, it holds that $\lambda_v\in\{\lambda,\lambda_1',\lambda_2'\}$. Moreover, if $\lambda_v\neq \lambda$, then $v$ has degree two in $G$.
\end{enumerate}
Using the bipartite graphs $G_1,G_2$ that implement $\lambda_1',\lambda_2'$ respectively, we obtain from Lemma~\ref{lem:nonuniform} that we can construct in linear time a bipartite graph $G'=(V_{G'},E_{G'})$ of maximum degree at most $\Delta$ such that
\[Z_{G'}(\lambda)=C_1^{n_1}C_2^{n_2}\cdot Z_G(\lambdab),\]
where $n_1,n_2$ are the number of vertices in $G$ whose activity equals $\lambda_1',\lambda_2'$, respectively. Note, the fact that  $G'$ is a bipartite graph whose maximum degree is at most $\Delta$ follows from the construction of Lemma~\ref{lem:nonuniform} and Item~\ref{it:ws234w} (see also Remark~\ref{rem:blowup}).

It follows that
\begin{equation}\label{eq:estimate23}
Z_{H,\beta,\gamma}= Z_{G'}(\lambda)/\big(C^{|E_H|}C_1^{n_1}C_2^{n_2}\big).
\end{equation}
Since the size of $G'$ is bigger than the size of $H$ only by a constant factor, there exists a constant $c'>1$ (depending only on $\lambda$) such that an approximation to $|Z_{G'}(\lambda)|$ within a multiplicative factor $(c')^{|V_{G'}|}$ yields via \eqref{eq:estimate23} an estimate to $|Z_{H,\beta,\gamma}|=Z_{H,\beta,\gamma}$ within a multiplicative factor $c^{|V_H|}$. It follows that $\HardCore{c'}$ is $\NP$-hard.
This completes the proof of Theorem~\ref{thm:neglambda}.
\end{proof}

\end{document}